\documentclass{article}
\usepackage{arxiv}

\usepackage[utf8]{inputenc}
\usepackage[T1]{fontenc}
\usepackage{amsmath,amssymb,amsthm}
\usepackage{graphicx,subcaption}
\usepackage{booktabs}
\usepackage{microtype}
\usepackage[numbers]{natbib}
\usepackage{hyperref}   
\hypersetup{colorlinks=true, linkcolor=blue, citecolor=blue, urlcolor=blue}

\theoremstyle{plain}

\newtheorem{theorem}{Theorem}[section]
\newtheorem{lemma}[theorem]{Lemma}
\newtheorem{proposition}{Proposition}

\theoremstyle{definition}
\newtheorem{definition}[theorem]{Definition}

\theoremstyle{remark}
\newtheorem{case}{Case}

\newcommand{\PSD}{f}
\newcommand{\vect}[1]{\boldsymbol{#1}}
\newcommand{\matr}[1]{\mathbf{#1}}

\newcommand{\Real}{\operatorname{Re}}
\newcommand{\Imag}{\operatorname{Im}}
\newcommand{\mean}{\mathbb E}
\newcommand{\vari}{\mathbb Var}

\title{Decomposing the spectral density function of an ARMA process}

\author{
  Darlin Soto-V\'{a}squez \\
  Departamento de Estad\'{\i}stica\\
  Universidad del B\'{\i}o-B\'{\i}o\\
  \texttt{dmsoto@ubiobio.cl}
  \And
  Francisco Cuevas-Pacheco \\
  Departamento de Matem\'{a}tica\\
  Universidad T\'{e}cnica Federico Santa Mar\'{i}a\\
  \texttt{francisco.cuevas@usm.cl}
  \AND
  Giovanni Motta \\
  Department of Mathematics and Statistics\\
  The City University of New York\\
  \texttt{gm444@hunter.cuny.edu}
  \And
  Malgorzata Sobolewska \\
  Harvard-Smithsonian Center for Astrophysics\\
  \\
  \texttt{msobolewska@cfa.harvard.edu}
}

\date{}

\begin{document}
\maketitle

\begin{abstract}
Quasi-periodic oscillations are characterized by fitting Lorentzian functions to a nonparametric estimate of the power spectral density, a two-step procedure that separates estimation from modeling and does not accommodate measurement error. We develop a parametric alternative based on autoregressive moving-average (ARMA) processes observed with measurement error. We derive an exact decomposition of the ARMA spectral density into real-valued components, one for each class of autoregressive roots under complex conjugation, plus a polynomial term when the moving-average order is at least the autoregressive order. Each component is a Wrapped Cauchy kernel, and for complex roots it is accompanied by a conjugate-Poisson term that makes the peak asymmetric and displaces its maximum from the frequency given by the argument of the root. We obtain closed-form expressions for the location, height, width, and asymmetry of each peak, and a matching variance decomposition. Parameters are estimated by maximum likelihood through the Kalman filter, and all descriptors follow by plug-in. Simulations show that the closed-form peak location is recovered without appreciable bias. An X-ray binary light curve illustrates the methodology.
\end{abstract}

\keywords{ARMA process \and Spectral density \and Measurement error}

\par\noindent{\bfseries\itshape MSC2020}\enspace Primary 62M15, 62M10; secondary 62P35, 85A35.

\section{Introduction}\label{intro}

Time series analysis is a fundamental tool in many scientific disciplines, from
astrophysics and geophysics to neuroscience and economics. A central objective is
to characterize the oscillatory behavior of a process, often by studying its power
spectral density (PSD). In particular, identifying and quantifying periodic or
quasi-periodic oscillations (QPOs) embedded in noisy, broadband signals is a
common challenge across fields. Traditionally, practitioners have relied on a
two-step approach: first, compute a nonparametric estimator of the PSD, such as
the periodogram or the Welch estimator \citep{Welch1967}; second, fit a sum of
phenomenological functions, typically Lorentzians, to the estimated spectrum to
extract peak frequencies, widths, and amplitudes. While widely adopted, this
methodology suffers from several statistical limitations: it separates estimation
from modeling, leading to loss of efficiency; it lacks a coherent probabilistic
framework for uncertainty quantification; and it does not naturally accommodate
measurement error or complex dependence structures.

An alternative and more principled approach is to represent the observed time
series using a parametric model that directly captures its second-order
structure. Among parametric models, autoregressive moving-average (ARMA)
processes \citep{Brockwell2001,brockwell1991} occupy a central role in time
series analysis due to their flexibility and parsimony. An ARMA($p,q$) process is
defined by a finite set of autoregressive and moving-average parameters, and its
spectral density has a known rational form. When the process is stationary and
causal, the shape of the PSD is governed by the roots of the autoregressive
polynomial. This suggests that, rather than first estimating the spectrum
nonparametrically and then fitting Lorentzians, one could estimate the ARMA
parameters directly and derive the spectral peaks analytically from those roots.

The idea of attributing spectral features to individual autoregressive roots is
not new. In continuous time, the spectral density of a continuous-time
autoregressive moving average (CARMA) process admits a decomposition into
Lorentzian-type components, one per autoregressive root, a property exploited by
\cite{Kelly2014} to analyze quasi-periodic oscillations in astrophysical light
curves. In discrete time, \cite{Pelikan84} expressed the ARMA spectral density as
a sum of root-associated terms. What is missing is a decomposition of the ARMA
spectral density into real-valued components, valid for arbitrary orders $p$ and
$q$, together with closed-form descriptors of the peaks those components produce.
The present paper fills this gap.

Our first contribution is an exact decomposition of the spectral density of a
causal ARMA($p,q$) process. We show that the PSD splits into $p$ components, one
for each root of $\phi(z^{-1})=0$, together with a polynomial term
$\PSD_{X,0}(\omega)$ that appears only when $q\ge p$ and which we characterize in
closed form. Individual components associated with complex roots are
complex-valued, so we group the roots into equivalence classes under conjugation
and show that each class contributes a real-valued function: a real root
contributes a single rescaled Wrapped Cauchy kernel, equivalently a rescaled
Poisson kernel, centered at frequency $0$ or at the Nyquist frequency
$1/(2\Delta)$ according to its sign, while a complex-conjugate pair contributes
the sum of a Wrapped Cauchy kernel and a sine-modulated conjugate-Poisson term,
together with their reflections about the origin. Writing $n_R$ and $n_C$ for the
numbers of real and complex classes, the spectral density is thus a sum of
$n_R+2n_C$ real-valued terms, plus one further term when $q\ge p$. The Wrapped
Cauchy kernel carries the entire spectral mass of a class, with concentration
governed by the modulus $s$ of the root, while the conjugate-Poisson term
integrates to zero and affects only the shape of the peak.

Our second contribution concerns the peaks themselves. In the phenomenological
approach, the oscillation frequency of a QPO is read off as the centroid
$\omega_{j,0}/(2\pi)$ of a Lorentzian fitted to the periodogram, and in the
CARMA representation of \cite{Kelly2014} this centroid is given by the imaginary
part of an autoregressive root. Our decomposition yields the discrete-time
analogue: each complex class carries a Wrapped Cauchy kernel centred at
$\alpha/(2\pi\Delta)$, where $\alpha$ is the argument of the representative root,
so that $\alpha/(2\pi\Delta)$ plays the role of the Lorentzian centroid. We show,
however, that this centroid is not where the class component peaks: the
conjugate-Poisson term is antisymmetric about $\alpha/(2\pi\Delta)$ and tilts the
kernel, so the maximum is displaced. We derive the exact location and height of
that maximum in closed form, together with the frequencies at which the
component falls to a fixed fraction $k$ of its peak value, the resulting width,
which at $k=1/2$ is the full width at half maximum, and an asymmetry index
measuring the imbalance between the two half-widths. All of these are explicit
functions of the model parameters, requiring no numerical search over frequency.
The displacement from $\alpha/(2\pi\Delta)$ is negligible for sharply
concentrated roots but substantial for broad ones, so reading the peak frequency
from the argument of the root, as the analogy with the Lorentzian centroid
suggests, is justified only in the former regime.

Our third contribution is a decomposition of the process variance that matches
the spectral decomposition term by term: a real class contributes $2\pi K(r)$, a
complex class contributes $4\pi\,\Real(K(r))$, and measurement error contributes
$\sigma^2_\eta$. Because these contributions are signed, they are interpreted as
components of the variance rather than as proportions, and they allow the
relative importance of broadband and oscillatory features to be quantified.
Throughout, the decomposition is developed for ARMA processes observed with
additive measurement error, which in applications sets the level of the
high-frequency floor and therefore the prominence of any peak above it.

Two strands of prior work are close to ours and help delimit our contribution.
The first is the frequency-domain decomposition of \cite{Pelikan84}, who obtained
$$\PSD_X(\omega) = \frac{1}{2\pi}\sum_{m=1}^{p} d_m\,
\frac{1 - r_m^2}{(1 - r_m e^{i 2\pi\omega \Delta})(1 - r_m e^{-i2\pi\omega \Delta})},$$
where $r_m$ are the roots of the autoregressive polynomial and $d_m$ are complex
coefficients obtained from the autocovariance function. This decomposition is
derived under three restrictions: the autoregressive polynomial must have
distinct roots; the autoregressive and moving-average polynomials must share no
common roots; and, critically, the autoregressive order must strictly exceed the
moving-average order, $p>q$. Under the last condition the moving-average roots do
not contribute, as the contour integral in the derivation is evaluated only at the
poles of the autoregressive polynomial. Moreover, the components are given in
complex form, without a separation into real-valued, individually interpretable
spectral contributions, and no descriptors of the resulting peaks are derived.

The second strand concerns continuous time and resonator-based models. The
spectral density of a CARMA$(p,q)$ process is a sum of Lorentzian-type terms, one
for each pair of complex-conjugate autoregressive roots and one for each real
root, whose centroid and width are determined by the imaginary and real parts of
the corresponding root \citep{Brockwell2001,Kelly2014}. Our decomposition plays
the same role in discrete time, but differs in three respects. First, the
discrete-time components are Wrapped Cauchy kernels, periodic on the frequency
band, rather than Lorentzians. As in continuous time when $q>0$, the complex
weights add an antisymmetric term that tilts the peak; we characterize the
resulting displacement, width and asymmetry in closed form, effects that the
reading of the peak frequency from the imaginary part of the root does not
account for. Second, a real root may be negative, producing a component centred
at the Nyquist frequency, which has no counterpart in continuous time. Third, a
CARMA$(p,q)$ process requires $q<p$, whereas our result treats arbitrary $q$
relative to $p$ and characterizes explicitly the polynomial term $\PSD_{X,0}$
that arises when $q\ge p$. A related line of work represents the spectrum
directly as a superposition of elementary resonators. In the Gaussian-process
literature, the celerite model of \cite{foremanmackey2017} builds covariance
functions as sums of stochastically-driven damped harmonic oscillators with
peaked, Lorentzian-type spectra, chosen so that likelihood evaluation scales
linearly in the sample size; and \cite{granadosgarcia2022} model a spectrum as a
Bayesian mixture of AR$(2)$ components, each contributing a single peak of given
location and bandwidth. These approaches posit resonators as modeling primitives
to be fitted, whereas we obtain the resonator structure as an exact algebraic
identity for the spectral density of an already-estimated ARMA$(p,q)$ process,
with each component tied to a specific autoregressive root.

Relative to \cite{Pelikan84}, we generalize in four ways. First, we remove the
restriction $p>q$, so that the decomposition covers all ARMA orders, with the
additional component $\PSD_{X,0}(\omega)$ characterized explicitly when $q\ge p$.
Second, we express the decomposition entirely through real-valued functions, each
tied to an equivalence class of autoregressive roots. Third, we derive the peak
location, width and asymmetry of each class in closed form, quantities that the
earlier decomposition does not provide and that are precisely what a
phenomenological Lorentzian fit is used to obtain. Fourth, we derive the
accompanying variance decomposition and extend the whole framework to processes
observed with measurement error.

To make the methodology operational, we estimate ARMA models by maximum
likelihood through a state-space representation and the Kalman filter
\citep{hagiwara2021time}, which accommodates measurement error naturally and
evaluates the likelihood at a cost linear in the sample size \citep{Gardner1980}.
Plug-in estimators of every spectral component and of every peak descriptor
follow by substituting the estimated parameters into the analytical expressions.
A simulation study covering an AR($2$) process with a single sharply concentrated
class, an AR($4$) process combining a sharp and a broad class, and an ARMA($3,2$)
process with one complex and one real class shows that the closed-form peak
location is recovered without appreciable bias whenever the underlying root is
identified. It also shows that the two natural alternatives are inadequate
outside the sharply concentrated regime: the nominal centroid
$\alpha/(2\pi\Delta)$ carries a bias equal to the displacement our result
predicts, which for a class of modulus $0.70$ is more than three times the
sampling standard deviation and does not diminish with the sample size; and a
numerical maximization of the fitted spectral density, which locates the maximum
of the total density rather than of the individual component, has a root mean
squared error an order of magnitude larger for the same class.

Finally, we illustrate the approach by analyzing an X-ray binary light curve from
the source XTE 1550$-$564, observed by the Rossi X-ray Timing Explorer (RXTE).
Following \cite{Kelly2014}, we apply a logarithmic transformation to
approximately Gaussianize the data and incorporate a known measurement error
variance estimated by the method of \cite{Vaughan2003}. The fitted ARMA($3,2$)
model captures the dominant quasi-periodic oscillation at approximately $2.84$~Hz
through a single complex-conjugate class, and the variance decomposition
attributes $55.5\%$ of the total variance to the broadband component associated
with the real class, $26.8\%$ to the oscillatory class, and $17.7\%$ to
measurement noise. Residual diagnostics support the adequacy of the fit.

While our application is motivated by astrophysics, the framework is general and
applicable wherever spectral peaks are of interest, including neuroscience,
geophysics and engineering. By providing an analytically tractable and
interpretable decomposition of the ARMA spectral density, together with
closed-form descriptors of the peaks it contains, this work connects parametric
time series modeling to the phenomenological peak fitting that is standard
practice in those fields.

The rest of the paper is organized as follows. Section~\ref{ARMA_models} reviews
ARMA processes with measurement error and their spectral density, and situates
the Lorentzian description of QPOs.
Section~\ref{decomp_PSD_ARMA} develops the spectral decomposition, the closed
forms for the peak location, width and asymmetry, and the variance decomposition.
Section~\ref{statespace_estimation} describes estimation by maximum likelihood
through the Kalman filter and the resulting plug-in estimators.
Section~\ref{simulations} reports the simulation study,
Section~\ref{application} the analysis of the RXTE light curve, and
Section~\ref{conclusion} concludes.

\section{Autoregressive moving-average process and spectral density}\label{ARMA_models}

This section develops the ARMA($p,q$) process with measurement error, which serves as the foundation for the modeling framework used throughout this work. We establish the model definition, its autocovariance structure, and its spectral density, which together provide the tools needed for the decomposition results that follow. Readers seeking a more exhaustive treatment of the classical theory may consult \cite{Brockwell2001} and \cite{brockwell1991}.

An ARMA$(p,q)$ process with measurement error \citep{liu2011}, denoted by $\{Y_t\}$, is a discrete-time real-valued stochastic process indexed by $t \in \mathbb{Z}$ with constant sampling interval $\Delta > 0$. It has autoregressive order $p$, moving average order $q$, autoregressive coefficient $\vect{\phi}=(\phi_1,\dots,\phi_p)$, moving average coefficient $\vect{\theta}=(\theta_1,\dots,\theta_q)$, and is driven by white noise $\{Z_t\}$. The process is defined by
\begin{align}
\phi(B) X_{t} &= \theta(B) Z_{t} \label{ARMA_def_sin_me}\\
Y_{t} &= X_{t} + \eta_{t}, \label{ARMA_def_with_me}
\end{align}
where $B$ denotes the backward shift operator, defined by $B X_t = X_{t-1}$, and the polynomials $\phi(\cdot)$ and $\theta(\cdot)$ are given by
\begin{align*}
    \phi(z) &= 1 - \phi_1 z - \cdots - \phi_p z^p, \\
    \theta(z) &= 1 + \theta_1 z + \cdots + \theta_q z^q,
\end{align*}
with $\phi_p \neq 0$ and $\theta_q \neq 0$. Here $\eta_{t}$ is the measurement error, $Y_{t}$ is the observed counterpart of $X_{t}$, and $Z_{t}$ and $\eta_{t}$ are mutually independent white noise sequences with zero means and variances $\sigma^2_Z$ and $\sigma^2_{\eta}$, respectively. We denote by $\vect{\beta}=(\vect{\phi},\vect{\theta},\sigma^2_Z)$ the vector of parameters of the underlying ARMA process $\{X_t\}$.

We assume throughout that the ARMA process is causal and stationary,
which requires every root of $\phi(z)$ to lie outside the unit circle,
i.e., $|z|>1$. Equivalently, the roots $r_1,\dots,r_p$ of
$\phi(z^{-1})$ satisfy $|r_j|<1$. We denote ARMA$(p,0)$ by AR$(p)$, and ARMA$(0,q)$ by MA$(q)$.

A causal ARMA($p,q$) process with measurement error has autocovariance function (ACVF) at lag $h$
\begin{align*}
\gamma_Y(0;\vect{\beta},\sigma^2_{\eta}) &= \gamma_X(0;\vect{\beta}) + \sigma^2_{\eta}, \\
\gamma_Y(h;\vect{\beta},\sigma^2_{\eta}) &= \gamma_X(h;\vect{\beta}), \qquad h \neq 0, \nonumber
\end{align*}
where the ACVF of $X_t$, $\gamma_X(h;\vect{\beta})$, is given by
\begin{equation}\label{ACVF_original}
\gamma_X(h;\vect{\beta}) = \sigma^2_{Z} \sum_{j=0}^{\infty} \xi_j\, \xi_{j+|h|}, \qquad h \in \mathbb{Z},
\end{equation}
with the coefficients $\xi_j$ obtained from the power series expansion
\begin{equation*}
\xi(z) = \sum_{j=0}^{\infty}\xi_j z^j = \frac{\theta(z)}{\phi(z)}, \qquad |z|\leq 1,
\end{equation*}
and $\xi_0=1$.

The spectral density of a causal ARMA($p,q$) process with measurement error is
\begin{align}
\PSD_Y(\omega;\vect{\beta},\sigma^2_{\eta}) &= \PSD_X(\omega;\vect{\beta}) + \PSD_{\eta}(\omega;\sigma^2_{\eta}), \qquad \frac{-1}{2\Delta} < \omega \leq \frac{1}{2\Delta} \label{PSD_ARMA_with_error}\\
\PSD_X(\omega;\vect{\beta}) &= \frac{\Delta \sigma^2_Z}{2\pi} \frac{\left| \theta(e^{-i2\pi \omega\Delta}) \right|^2}{\left| \phi(e^{-i2\pi \omega\Delta}) \right|^2} \label{PSD_ARMA_without_error}\\
\PSD_{\eta}(\omega;\sigma^2_{\eta}) &= \frac{\Delta \sigma^2_{\eta}}{2\pi}. \label{PSD_me}
\end{align}

Astronomers have long used the periodogram to study the PSD of astronomical sources \citep{PapadakisLawrence1993,Uttley2014,vanderKlis1989}, and in a variety of accreting systems this estimator reveals two recurring features: broad, low-frequency structures associated with noise, and narrow, finite-width peaks known as quasi-periodic oscillations. A QPO is usually modeled by a Lorentzian function (also known as the Cauchy distribution),
\begin{equation}\label{Lorenztian}
L(\omega) = \frac{\lambda}{(2\pi\omega-\omega_0)^2+(\lambda/2)^2},
\end{equation}
characterized by centroid frequency $\omega_0/(2\pi)$, full width at half maximum (FWHM) $\lambda/(2\pi)$, and quality factor $Q=\omega_0/\lambda$; signals are conventionally classified as QPOs when $Q>2$ \citep{vanderKlis2006}. Figure~\ref{fig:lorent_example} illustrates this parametrization, indicating the centroid frequency and FWHM of a Lorentzian function.

\begin{figure}[tb]
\centering
\includegraphics[width=10cm]{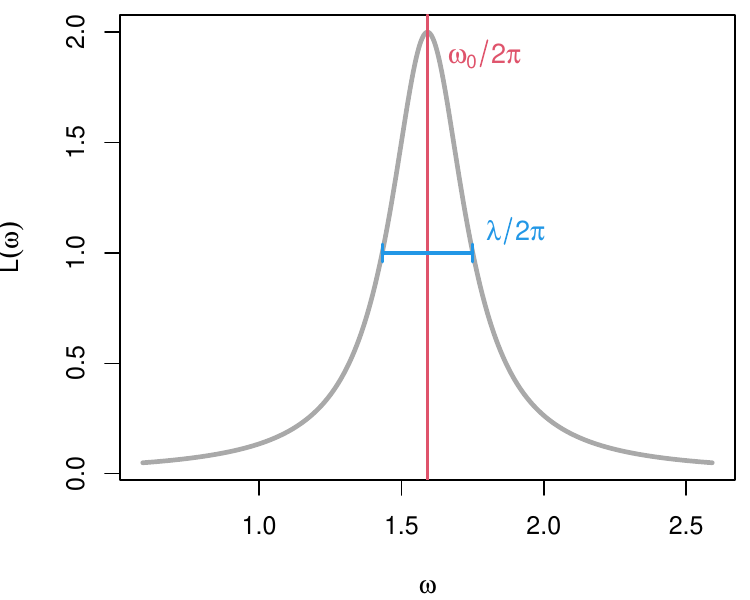}
\caption{Lorentzian function $L(\omega)$ in \eqref{Lorenztian}, with centroid at $\omega_0/2\pi=1.5915$ (red line) and FWHM equal to $\lambda/2\pi=0.3183$ (blue segment).}
\label{fig:lorent_example}
\end{figure}

Motivated by these observed peaks, a widely adopted approach assumes that the PSD itself can be represented as a sum of Lorentzian functions \citep{Belloni2002,Malzac2018,Nowak2000,pawar2015}:
\begin{equation*}
\PSD_X(\omega) = \sum_{j=1}^{J} L_j(\omega), \qquad
L_j(\omega)=\frac{\lambda_j}{(2\pi\omega-\omega_{j,0})^2+(\lambda_j/2)^2}, \quad j=1,\dots,J,
\end{equation*}
where each component $L_j$ is characterized by a centroid frequency $\omega_{j,0}/(2\pi)$ and a FWHM $\lambda_j/(2\pi)$. The parameters $(\omega_{j,0}, \lambda_j)$ are typically estimated by fitting this sum directly to the periodogram ordinates, via $\chi^2$ minimization or maximum likelihood \citep{Barret2012,Nowak2000}. This framework has been widely used to identify and classify QPOs and broadband noise components across accreting compact objects, most notably X-ray binaries \citep{IngramMotta2019}. Because the Lorentzian sum is fitted to an empirical, non-parametric spectral estimate rather than derived from an explicit generative model for the process, the resulting QPO parameters carry uncertainty inherited from the periodogram's sampling properties, without a direct link to a model-based likelihood for the underlying time series.

An alternative to this non-parametric route is to obtain the same Lorentzian decomposition directly from a parametric model of the process. In particular, \cite{Kelly2014} showed that the PSD of a CARMA($p,q$) process can be decomposed exactly as a sum of Lorentzian functions, one for each pair of complex-conjugate roots $r_j = a_j \pm ib_j$ of the autoregressive polynomial (plus an additional real-root term when $p$ is odd). Each root pair defines a Lorentzian $L_j(\omega)$ of the same form as above, whose centroid frequency is given by the imaginary part of the root and whose width is given by its real part, so that $\omega_{j,0}$ and $\lambda_j$, and hence the quality factor $Q_j=\omega_{j,0}/\lambda_j$, follow directly from the estimated model coefficients rather than from a separate fit to the periodogram. Unlike the periodogram-based fit, this route estimates the parameters of the CARMA process by maximizing the exact CARMA likelihood via the Kalman filter, so uncertainty in the QPO parameters is inherited from the model likelihood itself. As before, each centroid $\omega_{j,0}/(2\pi)$ is a candidate QPO frequency, with $Q_j$ determining whether it is retained as a genuine QPO or reflects broadband noise. Because the same root-decomposition argument applies to any process whose spectral density is a ratio of trigonometric polynomials, discrete-time ARMA models are, in principle, equally valid candidates for this Lorentzian-based QPO decomposition.

\section{Decomposition of the spectral density}\label{decomp_PSD_ARMA}

CARMA processes are the continuous-time analogs of the well-known discrete-time ARMA models, extending the classical framework to irregularly sampled data, as often encountered in astronomical observations. A key feature of CARMA models is that their PSD admits a closed-form decomposition as a sum of Lorentzian functions, each associated with a root of the autoregressive polynomial. This property makes CARMA modeling particularly effective at capturing QPOs \citep{Kelly2014}.

The following proposition establishes a discrete-time counterpart to this result: the PSD of a causal ARMA$(p,q)$ process can likewise be decomposed algebraically into a sum of
root-associated spectral components, one per root of the autoregressive polynomial, together with a residual polynomial term that arises whenever $q \geq p$. Although each component in the ARMA case is not an exact Lorentzian as in continuous time, it retains a similar peaked shape governed by the modulus and argument of the corresponding root, allowing the same modal interpretation of the spectral density used in CARMA modeling to be carried over to the discrete-time setting.

\begin{proposition}\label{prop:PSD_ARMA}
Let $\{X_t\}$ be a causal ARMA$(p,q)$ process defined in \eqref{ARMA_def_sin_me}, and suppose the
polynomial $\phi(z^{-1}) = 1 - \phi_1 z^{-1} - \cdots - \phi_p z^{-p}$ has $p$ simple
(distinct) roots, collected in $\mathcal{R}$. Then, the PSD given in
\eqref{PSD_ARMA_without_error} admits the decomposition
\begin{equation}\label{PSD_ARMA_sum}
\PSD_X(\omega;\vect{\beta})=\mathbb{I}_{\{q\geq p\}}\, \PSD_{X,0}(\omega;\vect{\beta})+\sum_{r \in \mathcal{R}} \PSD_{X,r}(\omega;\vect{\beta}),
\end{equation}
where $\mathbb{I}_{\{q\geq p\}}$ equals $1$ if $q\geq p$ and $0$ otherwise,
and functions $\PSD_{X,0}$ and $\PSD_{X,r}$ are given by
\begin{align}
    \PSD_{X,0}(\omega;\vect{\beta})&= \frac{\Delta\sigma_Z^2}{2\pi}\left[\, c_0 + 2\sum_{h=1}^{m} c_h \cos(2\pi h\,\omega\Delta)\right], \label{P0}\\
    \PSD_{X,r}(\omega;\vect{\beta})& = K(r)\,\frac{\Delta(1-r^2)}{1-2r \cos(2\pi \omega \Delta)+r^2},\label{Pr}
\end{align}
with $m=q-p$ and
\begin{equation*}
c_h = \frac{1}{(m-|h|)!}\left[\frac{d^{\,m-|h|}}{dz^{m-|h|}} \frac{z^{q-p}\theta(z)\theta(z^{-1})}{\phi(z)\phi(z^{-1})}\right]_{z=0}, \qquad |h| \le m,
\end{equation*}
and
\begin{equation}\label{Kr}
K(r) = -\frac{\sigma^2_Z}{2\pi}\,\frac{r\, \theta(r)\, \theta(r^{-1})}{\phi(r)\, \phi'(r^{-1})}.
\end{equation}
\end{proposition}
The proof is given in Appendix \ref{proof_prop1}. Proposition~\ref{prop:PSD_ARMA}
shows that the decomposition consists of one component $\PSD_{X,r}$ for each
autoregressive root $r$, together with the term $\PSD_{X,0}$, which is
present whenever $q \geq p$ and vanishes when $q<p$, in which case the
spectral density is entirely represented by the $p$ root-associated
components. The component $\PSD_{X,0}$ is a real-valued function with the
form of the PSD of an MA($m$) process, with $m=q-p$, although the
coefficients $c_h$ are not the autocovariances of a genuine MA$(m)$ process,
since they depend jointly on $\vect{\phi}$ and $\vect{\theta}$ rather than on $\vect{\theta}$
alone. In contrast, each $\PSD_{X,r}$ has the functional form of the PSD of an AR$(1)$
process, weighted by $K(r)$, and is complex-valued or real-valued depending on whether the root
$r$ is complex or real. The moving-average polynomial affects the weights of
the root-associated components through the factors $\theta(r)$ and
$\theta(r^{-1})$, but does not change their number, which is always $p$. Although $\theta(z^{-1})$ and $\phi(z^{-1})$ are singular at the origin, the
ratio defining $c_h$ is not: writing $\theta^{*}(z)=z^{q}\theta(z^{-1})$ and
$\phi^{*}(z)=z^{p}\phi(z^{-1})$, the factor $z^{q-p}$ cancels the poles and
leaves $\theta(z)\theta^{*}(z)/\big(\phi(z)\phi^{*}(z)\big)$, which is analytic
at $z=0$ because $\phi(0)=1$ and $\phi^{*}(0)=-\phi_p\neq0$.

The assumption that the roots of $\phi(z^{-1})$ are simple is generic, in the sense
that it fails only on a set of autoregressive coefficients of Lebesgue measure zero,
and it is what allows the partial-fraction expansion to produce one first-order term
per root. If a root $r$ has multiplicity $\nu>1$, its contribution to the PSD is no
longer, in general, of the AR(1)-type form given in \eqref{Pr}, as additional terms depending
on the multiplicity of the root arise. We do not pursue this case here. It should be noted, however, that
$K(r)$ in \eqref{Kr} is singular when two roots coalesce, since
$\phi'(r^{-1})\to0$, so the individual weights become numerically unstable for models
with nearly repeated roots even though the total spectral density remains
well behaved.

In Proposition~\ref{prop:PSD_ARMA}, the components $\PSD_{X,r}(\omega;\vect{\beta})$ in \eqref{Pr} are indexed by the roots $r\in\mathcal{R}$, where $\mathcal{R}$ denotes the set of roots of $\phi(z^{-1})=0$, which may be real or complex. For real roots, $\PSD_{X,r}$ is real-valued and can be interpreted directly as a root-associated spectral contribution. In contrast, for complex roots, individual terms are generally complex-valued and do not constitute spectral densities on their own. Since $\phi$ has real coefficients, complex roots occur in conjugate pairs, whose corresponding terms combine into a real-valued contribution. We therefore group the roots according to complex conjugation
\begin{equation*}
    r \sim r' \iff r' \in \{r, \bar{r}\}.
\end{equation*}
\noindent We write $\mathcal{R}/\!\sim$ for the set of equivalence classes and denote a class by $\tau$. Each equivalence class consists either of a single real root or of a complex-conjugate pair. For each class, we choose the representative $r$ with nonnegative imaginary part.

Throughout, we identify each class $\tau\in\mathcal{R}/\!\sim$ with the polar representation of its representative root: we write $\tau=(s,\alpha)$, where $s=|r|\in(0,1)$ and $\alpha=\arg(r)\in[0,\pi]$ for the representative $r$ with nonnegative imaginary part. Thus $s$ and $\alpha$ always refer to the class under consideration, and in sums indexed by $\tau$ they vary with the summand.

As an example, consider an ARMA($3,2$) process with roots
$\mathcal{R} = \{0.5,\; 0.7\,e^{\pm i\pi/4}\}$, that is, one real root and one
complex-conjugate pair. Complex conjugation groups these into two classes,
$\tau_1 = \{0.5\}$ and $\tau_2 = \{0.7\,e^{\pm i\pi/4}\}$, so that
$\mathcal{R}/\!\sim\, = \{\tau_1, \tau_2\}$. The real class $\tau_1$ contributes
the single component $\PSD_{X,0.5}$, while the complex class $\tau_2$ contributes
the real sum of its two conjugate terms.

Grouping the terms of \eqref{PSD_ARMA_sum} by class rather than by root, and
writing $\widetilde{\PSD}_{X,\tau}$ for the contribution of the class $\tau$,
the decomposition becomes

\begin{equation*}
    \PSD_X(\omega;\vect{\beta}) = \mathbb{I}_{\{q \geq p\}}\, \PSD_{X,0}(\omega;\vect{\beta})
    + \sum_{\tau \in \mathcal{R}/\sim} \widetilde{\PSD}_{X,\tau}(\omega;\vect{\beta}),
\end{equation*}
\noindent where $\widetilde{\PSD}_{X,\tau}(\omega;\vect{\beta}) = \sum_{r \in \tau}\PSD_{X,r}(\omega;\vect{\beta})$
is real-valued for every class, by the argument above.

Naturally, the expression of $\widetilde{\PSD}_{X,\tau}$ depends on the nature of the roots,
and we distinguish two cases.

\begin{case}\label{case:complex}
When a class $\tau$ corresponds to a complex-conjugate pair, i.e. $\tau = \{r, \bar{r}\}$, $r = a + ib = s\,e^{i\alpha}$,
with $s = \sqrt{a^2 + b^2} = |r| \in (0,1)$ and $\alpha = \arg(r)$, the two
roots combine into the real mode:
\begin{equation}\label{Pcomplex}
    \widetilde{\PSD}_{X,\tau}(\omega;\vect{\beta}) = \PSD_{X,r}(\omega;\vect{\beta}) + \PSD_{X,\bar{r}}(\omega;\vect{\beta}) = \psi_{\tau}(\omega) + \psi_{\tau}(-\omega)
\end{equation}
where
\begin{equation}\label{Pcomplex_LS}
\psi_{\tau}(\omega) = \varphi_{s}(\omega,\alpha) + \zeta_{s}(\omega,\alpha)
\end{equation}
and
\begin{align}
    \varphi_{s}(\omega,\alpha)&= \Real \big(K(r)\big)\, \frac{\Delta(1-s^2)}{1-2s\cos(2\pi \omega \Delta - \alpha)+s^2}, \label{varphi} \\
    \zeta_{s}(\omega,\alpha)&= \Imag\big(K(r)\big)\, \frac{\Delta\, 2s \sin(2\pi \omega\Delta - \alpha)}{1-2s\cos(2\pi \omega \Delta - \alpha) + s^2}.\label{zeta}
\end{align}

The proof of \eqref{Pcomplex} is given in Appendix~\ref{proof_P_complex}. Since
$\varphi_{s}$ is even and $\zeta_{s}$ is odd in $2\pi\omega\Delta-\alpha$, and the
weights $\Real\big(K(r)\big)$ and $\Imag\big(K(r)\big)$ in \eqref{varphi} and
\eqref{zeta} are held fixed, the reflected term satisfies
$\psi_{\tau}(-\omega)=\varphi_{s}(\omega,-\alpha)-\zeta_{s}(\omega,-\alpha)$.
Equivalently, since $\Real\big(K(\bar r)\big)=\Real\big(K(r)\big)$ and
$\Imag\big(K(\bar r)\big)=-\Imag\big(K(r)\big)$, the reflected term is obtained
from $\psi_{\tau}$ by replacing $\alpha$ with $-\alpha$ and the weights of $K(r)$
with those of $K(\bar r)$, so that the conjugate root contributes the mirror image
of $\psi_\tau$ about the origin.

The term $\psi_{\tau}(\omega)$ decomposes as a weighted sum of a Wrapped Cauchy kernel $\varphi_{s}(\omega,\alpha)$, also known as a Wrapped Lorentzian function or rescaled Poisson kernel \citep{Mardia1999}, and a conjugate-Poisson-type term $\zeta_{s}(\omega,\alpha)$ \citep{Katznelson_2004}, both centred at the frequency $\alpha/(2\pi\Delta)$. The kernel $\varphi_{s}(\omega,\alpha)$, weighted by $\Real\big(K(r)\big)$, determines the principal spectral concentration, while $\zeta_{s}(\omega,\alpha)$, weighted by $\Imag\big(K(r)\big)$, supplies the signed correction required to reproduce the exact root contribution. The frequency $\alpha/(2\pi\Delta)$ is therefore a natural candidate for the centroid of the oscillatory component associated with $\tau$, in analogy with the centroid of the Lorentzian profiles commonly fitted to X-ray binary power spectra. However, this centroid does not in general coincide with the peak of the spectral component, because the term $\zeta_{s}(\omega,\alpha)$ is antisymmetric about $\alpha/(2\pi\Delta)$, so it skews the symmetric kernel $\varphi_{s}(\omega,\alpha)$ and shifts the maximum of $\psi_{\tau}(\omega)$. Consequently, $\alpha/(2\pi\Delta)$ alone does not determine the peak frequency; the physically meaningful quantities, namely the oscillation frequency and width, must be obtained from $\widetilde{\PSD}_{X,\tau}(\omega;\vect{\beta})$ or from the full spectral density.

The kernel $\varphi_{s}(\omega,\alpha)$ is unimodal and symmetric about its center, with $s$ controlling its concentration: as $s\to 0$, the Wrapped Cauchy kernel flattens toward a constant (uniform) spectral contribution, while as $s\to 1$, it becomes sharply peaked at its corresponding center frequency. Being non-negative and symmetric about its center, the Wrapped Cauchy kernel carries the entire spectral mass of the term and is normalized over the frequency domain, satisfying
\begin{equation}\label{int_Ldes}
\int_{-\frac{1}{2\Delta}}^{\frac{1}{2\Delta}}
\frac{\Delta(1-s^2)}{1-2s\cos(2\pi \omega \Delta \pm \alpha)+s^2}\, d\omega = 1.
\end{equation}

In contrast, $\zeta_{s}(\omega,\alpha)$ is antisymmetric about its center and integrates to zero
\begin{equation}\label{int_S}
\int_{-\frac{1}{2\Delta}}^{\frac{1}{2\Delta}} \frac{\Delta\, 2s \sin(2\pi \omega\Delta - \alpha)}{1-2s\cos(2\pi \omega \Delta - \alpha) + s^2} d\omega = 0.
\end{equation}
\end{case}

\begin{case}\label{case:real}
When the root $r$ is real, the class $\tau$ contains a single element, and the mode reduces to a single root-associated component, $\widetilde{\PSD}_{X,\tau}(\omega;\vect{\beta})=\PSD_{X,r}(\omega;\vect{\beta})$, given by the weighted Wrapped Cauchy kernel in \eqref{varphi} with concentration $s=|r|$ and mean direction (centroid) at frequency $0$ or at the Nyquist frequency $1/(2\Delta)$, depending on the sign of $r$, that is,
\begin{equation}\label{aux5}
\widetilde{\PSD}_{X,\tau}(\omega;\vect{\beta})=\PSD_{X,r}(\omega;\vect{\beta}) =
\begin{cases}
\varphi_{s}(\omega, 0), & r > 0, \\
\varphi_{s}(\omega, \pi), & r < 0.
\end{cases}
\end{equation}
\end{case}

Overall, the decomposition provides a direct interpretation of the autoregressive roots in the frequency domain. In particular, the modulus of a root controls the concentration of its associated spectral kernel,
whereas its argument determines the center frequencies of the corresponding oscillatory contributions. Considering Cases~1 and~2, the decomposition of $\PSD_X$ can be summarized as
\begin{align}
\PSD_X(\omega;\vect{\beta}) &= \mathbb{I}_{\{q\geq p\}}\, \PSD_{X,0}(\omega;\vect{\beta})
             + \sum_{\substack{\tau \in \mathcal{R}/\sim \\ |\tau|=1}} \widetilde{\PSD}_{X,\tau}(\omega;\vect{\beta})
             + \sum_{\substack{\tau \in \mathcal{R}/\sim \\ |\tau|=2}} \widetilde{\PSD}_{X,\tau}(\omega;\vect{\beta}) \notag\\
        &= \mathbb{I}_{\{q\geq p\}}\, \PSD_{X,0}(\omega;\vect{\beta})
             + \sum_{\substack{\tau \in \mathcal{R}/\sim \\ |\tau|=1}} \varphi_{s}(\omega, \alpha)
             + \sum_{\substack{\tau \in \mathcal{R}/\sim \\ |\tau|=2}} \big[\psi_{\tau}(\omega) + \psi_{\tau}(-\omega)\big] \label{PSD_decomp4}
\end{align}
where each equivalence class $\tau \in \mathcal{R}/\!\sim$ is either a singleton $\tau=\{r\}$ (a real class, $|\tau|=1$) or a conjugate pair $\tau=\{r,\bar r\}$ (a complex class, $|\tau|=2$), with $r = s e^{i\alpha}$.
When $|\tau|=1$, $\widetilde{\PSD}_{X,\tau}$ is given by \eqref{aux5}, whereas when $|\tau|=2$, $\widetilde{\PSD}_{X,\tau}$ is given by \eqref{Pcomplex}.
Denoting by $n_R$ and $n_C$, respectively, the number of real and complex
classes in $\mathcal{R}/\!\sim$, the total number of terms in $\PSD_X(\omega;\vect{\beta})$
is thus $1+n_R+2n_C$ when $q\geq p$, and $n_R+2n_C$ when $q<p$. Unlike the
decomposition in \eqref{PSD_ARMA_sum}, where each $\PSD_{X,r}$ may be
complex-valued when $r$ is a complex root, all components in
\eqref{PSD_decomp4} are real-valued functions.

\begin{figure}[!t]
\centering
\includegraphics[width=12cm]{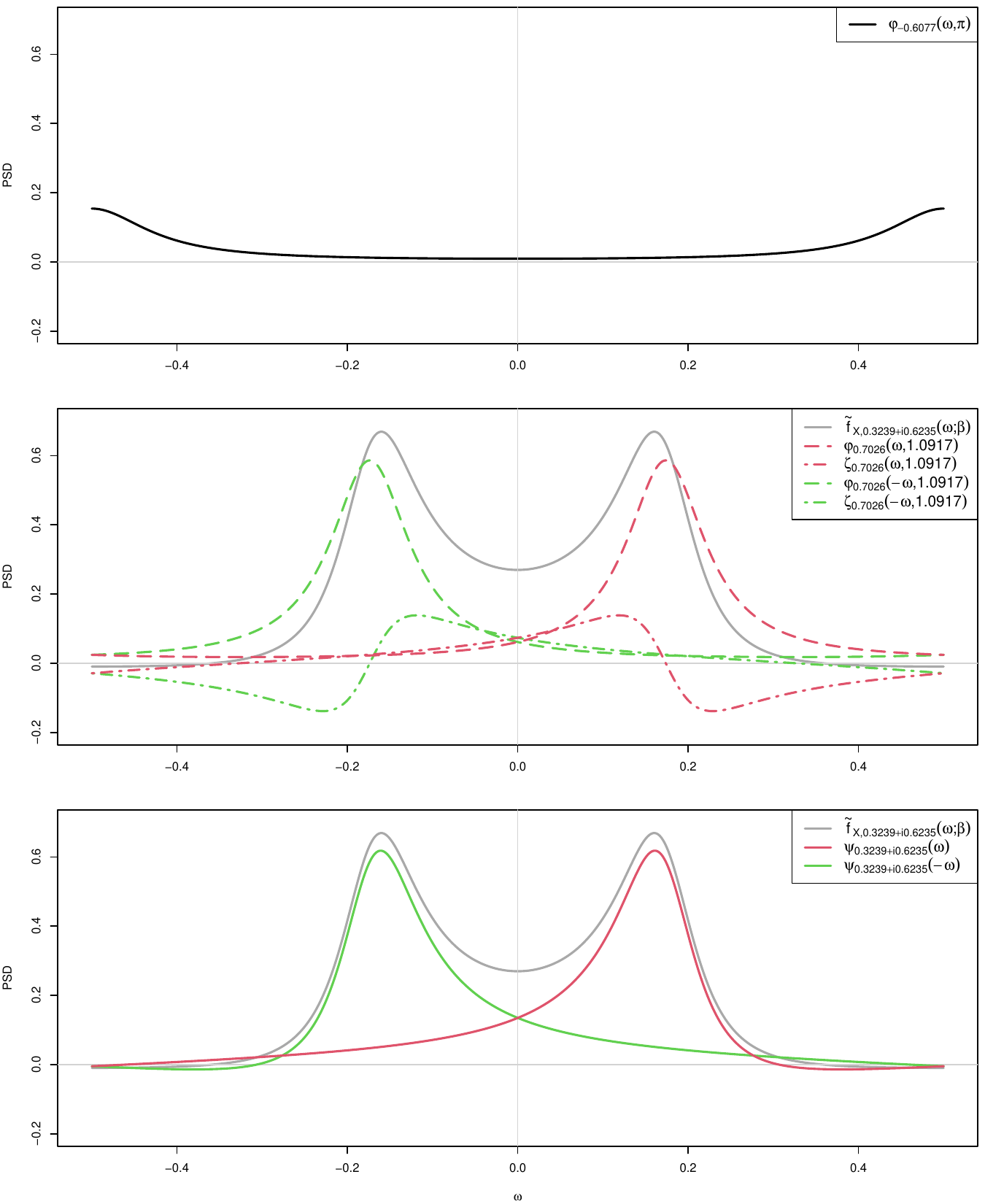}
\caption{Spectral density components for an autoregressive polynomial with one
real root and a complex-conjugate pair. Top panel: the real-root component
$\widetilde{\PSD}_{X,\tau}(\omega;\vect{\beta})$ for $\tau=\{-0.6077\}$ ($s=0.6077$,
$\alpha=\pi$; black). Middle panel: for the complex-conjugate class
$\tau=\{0.3239\pm i0.6235\}$ ($s=0.7026$, $\alpha=1.0917$), the solid gray line is
$\widetilde{\PSD}_{X,\tau}(\omega;\vect{\beta})$, while the dashed and
dot-dashed lines show the kernels $\varphi_s$ and $\zeta_s$ evaluated at
$\alpha=1.0917$ (red) and at $-\alpha=-1.0917$ (green). Bottom panel: the solid
gray line again shows $\widetilde{\PSD}_{X,\tau}(\omega;\vect{\beta})$, while
the red and green solid lines show the two real-valued terms
$\psi_\tau(\omega)=\varphi_s(\omega,\alpha)+\zeta_s(\omega,\alpha)$ and
$\psi_\tau(-\omega)=\varphi_s(\omega,-\alpha)-\zeta_s(\omega,-\alpha)$,
illustrating how they sum to $\widetilde{\PSD}_{X,\tau}(\omega;\vect{\beta})$.}
\label{fig:spectral_components}
\end{figure}

Figure~\ref{fig:spectral_components} illustrates the decomposition results of Case~1 and Case~2 above, using an example with one real root, $\tau_1 = \{-0.6077\}$, and one complex-conjugate pair, $\tau_2 = \{0.3239 \pm i0.6235\}$ ($s = 0.7026$, $\alpha = 1.0917$). Neither the real-root component $\widetilde{\PSD}_{X,\tau_1}(\omega;\vect{\beta})$ nor the complex-pair component $\widetilde{\PSD}_{X,\tau_2}(\omega;\vect{\beta})$ is individually guaranteed to be non-negative; non-negativity holds only for the total spectral density, obtained by summing over all classes. The individual terms $\varphi_s(\omega,\pm\alpha)$ and $\zeta_s(\omega,\pm\alpha)$ underlying $\widetilde{\PSD}_{X,\tau_2}(\omega;\vect{\beta})$ likewise take negative values over part of the frequency range.

\begin{figure}[!t]
\centering
\includegraphics[width=13cm]{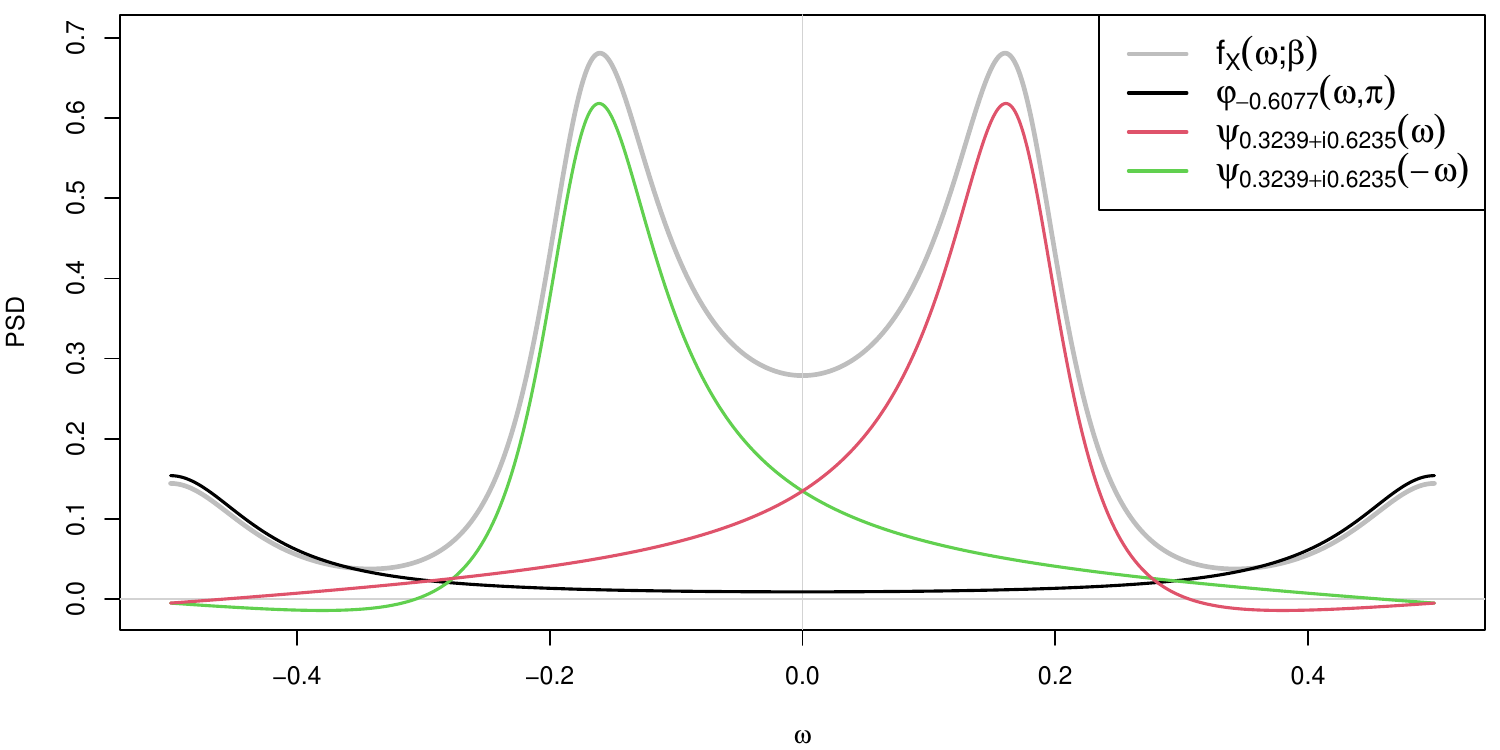}
\caption{Total PSD $\PSD_X(\omega;\vect{\beta})$ (gray line) of an ARMA(3,2) process with coefficients $\phi_1=0.04$, $\phi_2=-0.1$, $\phi_3=-0.3$, $\theta_1=0.5$, $\theta_2=0.3$, $\sigma^2_Z=1$ and $\Delta=1$, decomposed into the real-root component $\varphi_{-0.6077}(\omega,\pi)$ (black line) and the complex-conjugate-pair components $\psi_{0.3239+i0.6235}(\omega)$ (red) and $\psi_{0.3239+i0.6235}(-\omega)$ (green), as described in Figure~\ref{fig:spectral_components}.}
\label{fig:spectral_total}
\end{figure}

Figure~\ref{fig:spectral_total} shows how the components associated with $\tau_1=\{-0.6077\}$ and $\tau_2=\{0.3239\pm i0.6235\}$ combine to yield the total PSD $\PSD_X(\omega;\vect{\beta})$ of the ARMA($3,2$) process without measurement error with coefficients $\phi_1=0.04$, $\phi_2=-0.1$, $\phi_3=-0.3$, $\theta_1=0.5$, $\theta_2=0.3$, $\sigma_Z^2=1$, and $\Delta=1$. Although the individual terms $\psi_{0.3239+i0.6235}(\pm\omega)$ take negative values over parts of the frequency range, their sum with the real-root component $\varphi_{-0.6077}(\omega,\pi)$ recovers a non-negative total density. The two peaks in $\PSD_X(\omega;\vect{\beta})$ arise directly from $\psi_{0.3239+i0.6235}(\pm\omega)$, while the real-root component contributes a broader, low-frequency-dominated background.

Analogous to the centroid frequency in the classical Lorentzian functions, we aim to find the frequency peak of $\psi_\tau(\omega)$ in order to propose candidates to find a QPO. The following proposition shows that we can perform the classical analysis with the decomposition \eqref{PSD_decomp4}.

\begin{proposition}\label{prop:peak}
Let $\tau=\{r,\bar r\}$ be a complex-conjugate class with $r=s\,e^{i\alpha}$, and let
$\psi_\tau(\omega)$ be as in \eqref{Pcomplex_LS}. Also, let $A=\Real\big(K(r)\big)$,
$B=\Imag\big(K(r)\big)$ and $|K|=\sqrt{A^2+B^2}$, and assume $K(r)\neq0$. Then the
range of $\psi_\tau$ over $(-1/(2\Delta),1/(2\Delta)]$ is $[G_{\min},G_{\max}]$, where
\begin{equation}\label{eq:extremes}
G_{\max} = \Delta\,\frac{A(1+s^2)+2s|K|}{1-s^2},
\qquad
G_{\min} = \Delta\,\frac{A(1+s^2)-2s|K|}{1-s^2},
\end{equation}
each attained at a single frequency, namely
\begin{align}
\omega_{\max} &= \frac{1}{2\pi\Delta}\big[\alpha+\lambda^\star\big],
\qquad \lambda^\star=\operatorname{atan2}\big(B,\,G_{\max}/\Delta\big), \label{eq:wmax}\\
\omega_{\min} &= \frac{1}{2\pi\Delta}\Big[\alpha+\operatorname{atan2}\big(-(1-s^2)B,\;2s|K|-(1+s^2)A\big)\Big], \label{eq:wmin}
\end{align}
both reduced modulo $1/\Delta$ into $(-1/(2\Delta),1/(2\Delta)]$. In particular
$\psi_\tau$ is unimodal on that interval.
\end{proposition}

The proof is given in Appendix \ref{proof_prop2}.

A second parameter of interest is the width of the peak, which describes how
sharply $\omega_{\max}$ is resolved and which, for the Lorentzian profiles used
to describe QPOs, is the full width at half maximum. The following proposition
shows that, as with $\omega_{\max}$, this width admits a closed-form expression
in terms of the ARMA parameters $\vect{\beta}$.

\begin{proposition}\label{prop:length}
Let $A$, $B$, $|K|$, $G_{\max}$ and $G_{\min}$ be as in Proposition~\ref{prop:peak}, assume
$G_{\max}>0$, and set $k_{\min}=G_{\min}/G_{\max}$. For $k\in\big(\max\{0,k_{\min}\},1\big)$ put
$\ell_k=kG_{\max}$ and
\begin{equation}\label{eq:arc}
\mu_k=\operatorname{atan2}\big(B,\,\ell_k/\Delta\big),
\qquad
\gamma_k=\arccos\!\left(\frac{(\ell_k/\Delta)(1+s^2)-A(1-s^2)}{2s\sqrt{(\ell_k/\Delta)^2+B^2}}\right)\in[0,\pi].
\end{equation}
Then $\psi_\tau$ attains the level $\ell_k$ at exactly two frequencies,
\begin{equation}\label{eq:wLR}
\omega^{L,R}_k = \frac{1}{2\pi\Delta}\big[\alpha+\mu_k\mp\gamma_k\big],
\end{equation}
the set $\{\omega:\psi_\tau(\omega)\ge\ell_k\}$ is the single arc joining
them, it contains $\omega_{\max}$, and the total width is
\begin{equation}\label{eq:width}
W_k=\omega^R_k-\omega^L_k=\frac{\gamma_k}{\pi\Delta}.
\end{equation}
\end{proposition}

The proof is given in Appendix \ref{proof_prop3}. Together, Propositions~\ref{prop:peak} and \ref{prop:length} give a fully analytical characterization of the complex-conjugate component $\psi_\tau(\omega)$ associated with the class $\tau=\{r,\bar r\}$. Proposition~\ref{prop:peak} provides the location and height of both the peak, $(\omega_{\max}, G_{\max})$, and the trough, $(\omega_{\min}, G_{\min})$, while Proposition~\ref{prop:length} provides the crossing frequencies $\omega^{L,R}_k$. All of these quantities are obtained in closed form directly from the ARMA parameters $\vect{\beta}$, without requiring any numerical search over $\omega$.

Beyond the peak location and height, the pair $(\omega^L_k, \omega^R_k)$ from Proposition~\ref{prop:length} also yields a natural measure of width and asymmetry for the component $\psi_\tau(\omega)$. At level $k$, the width of the peak is $W_k$ as given in \eqref{eq:width}; in the classical case $k=1/2$, this reduces to the FWHM, the standard width parameter used to characterize the Lorentzian profiles employed to describe QPOs in astronomy \citep{Belloni2002,Nowak2000,vanderKlis1989}. The following definition, based on \cite{Heckman1981}, formalizes an asymmetry index built from the half-widths on either side of $\omega_{\max}$.

\begin{definition}\label{def:asym}
Let $\omega^L_k$ and $\omega^R_k$ be the left and right frequencies around $\omega_{\max}$ such that
\begin{equation*}
\psi_\tau(\omega^L_k) = \psi_\tau(\omega^R_k) = k\,\psi_\tau(\omega_{\max}),
\end{equation*}
and let $W^L_k = \omega_{\max} - \omega^L_k$ and $W^R_k = \omega^R_k - \omega_{\max}$ denote the corresponding left and right half-widths. The asymmetry index at level $k$ is
\begin{equation*}
AI_k = \frac{W^L_k - W^R_k}{W^L_k + W^R_k}.
\end{equation*}
\end{definition}

Since $\omega^L_k$ and $\omega^R_k$ are already available in closed form from Proposition~\ref{prop:length}, $AI_k$ inherits a closed-form expression directly in terms of the ARMA parameters $\vect{\beta}$, requiring no numerical evaluation of $\psi_\tau(\omega)$. Because $W^L_k, W^R_k > 0$, the index is bounded, $AI_k \in (-1,1)$, with $AI_k = 0$ corresponding to a symmetric peak, $W^L_k = W^R_k$; this occurs, in particular, when the root of the autoregressive polynomial $\phi(z^{-1})=0$ associated with the component is real. Positive (negative) values of $AI_k$ indicate that the peak decays more slowly toward lower (higher) frequencies than toward higher (lower) ones, i.e., a left- (right-) skewed QPO profile.

\begin{proposition}\label{prop:AIclosed}
Under the hypotheses of Proposition~\ref{prop:length}, with $\gamma_k>0$,
\begin{equation}\label{eq:AIclosed}
AI_k=\frac{\lambda^\star-\mu_k}{\gamma_k}
=\frac{1}{\gamma_k}\arctan\!\left(\frac{-(1-k)\,B\,(G_{\max}/\Delta)}{k\,(G_{\max}/\Delta)^2+B^2}\right).
\end{equation}
In particular $\operatorname{sign}(AI_k)=-\operatorname{sign}(B)$, and $AI_k\equiv0$ for
all $k$ if and only if $B=0$. Moreover $AI_k=O(\sqrt{1-k})$ as $k\uparrow1$.
\end{proposition}

The proof is given in Appendix \ref{proof_prop4}. The first equality in
\eqref{eq:AIclosed} identifies $AI_k$ as the angular offset of the peak from the
centre of the level arc, normalised by the arc half-width, and the rate
$O(\sqrt{1-k})$ explains why the index vanishes as the level approaches the
maximum: near its peak the component is locally quadratic and therefore locally
symmetric.

In numerical work, the three angles $\mu_k\pm\gamma_k$ and $\lambda^\star$ are returned
in different periods by independent $\operatorname{atan2}$ and $\arccos$ evaluations, so
the half-widths must be computed as
$W^L_k=\big[(\lambda^\star-\mu_k+\gamma_k)\bmod2\pi\big]/(2\pi\Delta)$ and
$W^R_k=W_k-W^L_k$. Because the
superlevel set is a single arc, this reduction is unambiguous and always places
$\omega_{\max}$ between the two crossings. When the arc contains the Nyquist frequency,
the two crossings appear at opposite ends of the plotted band.

When $B=0$ the component is the weighted Wrapped Cauchy kernel in \eqref{aux5}, and the results
above remain valid rather than degenerating. For $A>0$ one has $G_{\max}=\Delta A(1+s)/(1-s)$
and $\omega_{\max}=\alpha/(2\pi\Delta)$ with $\alpha\in\{0,\pi\}$, so the peak sits at
zero or at the Nyquist frequency, and $\mu_k=0$ for every $k$, whence $AI_k\equiv0$ by
Proposition~\ref{prop:AIclosed}: the peak is symmetric at every level, as it must be.
The width is $W_k=\gamma_k/(\pi\Delta)$ with
$\gamma_k=\arccos\big[(1+s^2)/(2s)-(1-s)^2/(2ks)\big]$, and $k_{\min}=(1-s)^2/(1+s)^2$,
so the level $kG_{\max}$ is crossed at all only for $k\ge(1-s)^2/(1+s)^2$. In particular the
FWHM of a real-root component exists only when $s\ge3-2\sqrt2\approx0.1716$; for
smaller $s$ the kernel is too weakly peaked to fall to half its maximum within the band.
For $A<0$ the maximum is $G_{\max}=\Delta A(1-s)/(1+s)<0$ and $k\in(0,1)$ is no longer a
meaningful parametrisation of the level.

\begin{figure}[h!]
\centering
\includegraphics[width=13cm]{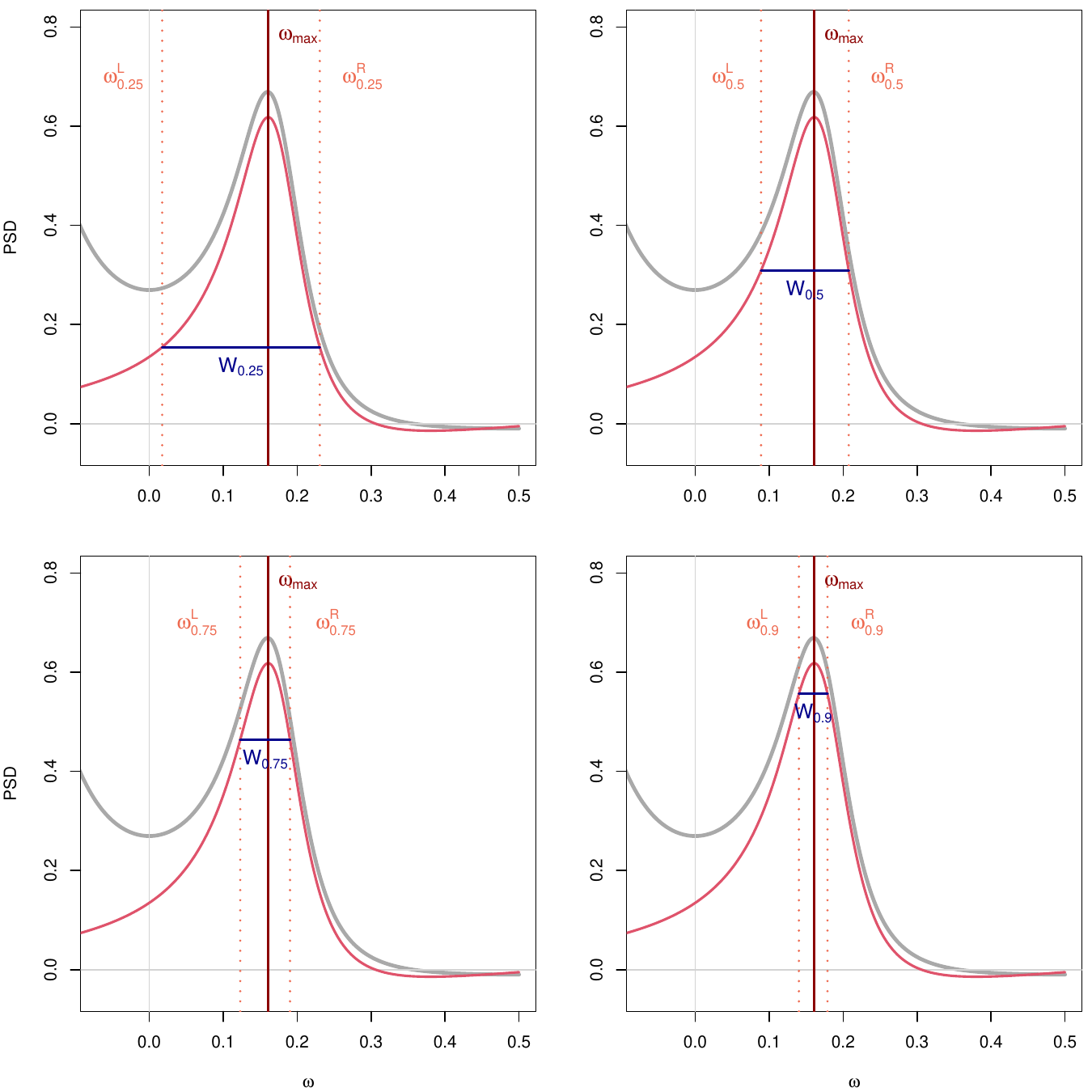}
\caption{Illustration of the peak frequency $\omega_{\max}$, boundary frequencies $\omega^L_k$ and $\omega^R_k$, and width $W_k=\omega^R_k-\omega^L_k$ for the component $\psi_{\tau}(\omega)$ given in \eqref{Pcomplex_LS} (red line), superimposed on the component $\widetilde{\PSD}_{X,\tau}(\omega;\vect{\beta})$ (gray line) with $\tau=\{0.3239 \pm i 0.6235\}$ for different values of $k$. The vertical solid line marks $\omega_{\max}$, the peak frequency of $\psi_{\tau}(\omega)$, while the dotted lines mark $\omega^L_k$ and $\omega^R_k$, the frequencies at which the component falls to a fraction $k$ of its peak value on either side of $\omega_{\max}$, and the horizontal segment indicates the resulting width $W_k$. Panels correspond to $k=0.25$, $0.5$, $0.75$, and $0.9$.}
\label{fig:width_example}
\end{figure}
Figure~\ref{fig:width_example} illustrates Propositions~\ref{prop:peak} and \ref{prop:length} for the component $\psi_\tau(\omega)$ associated with $\tau = \{0.3239 \pm i0.6235\}$ in the ARMA(3,2) example of Figure~\ref{fig:spectral_total}, superimposed on $\widetilde{\PSD}_{X,\tau}(\omega;\vect{\beta})$. Across all panels, the peak location $\omega_{\max}=0.1609$ and the reference level $G_{\max}=0.6181$ remain fixed, since both quantities depend only on the class $\tau$ and not on the level $k$. Although not displayed in the figure, Proposition~\ref{prop:peak} also gives the minimum of the component, $G_{\min}=-0.014$, attained at $\omega_{\min}=0.3801$; hence $\psi_\tau$ takes only slightly negative values away from the peak. As $k$ increases from $0.25$ to $0.9$, the threshold $\ell_k$ increases monotonically from $0.1545$ to $0.5563$, and the crossing frequencies $\omega^L_k$ and $\omega^R_k$ close in around $\omega_{\max}$, though at unequal rates on the two sides: for $k=0.25$, $\omega^L_{0.25}=0.0176$ and $\omega^R_{0.25}=0.2306$, whereas for $k=0.9$, $\omega^L_{0.9}=0.1402$ and $\omega^R_{0.9}=0.1788$. Correspondingly, the total width $W_k$ shrinks from $0.2130$ at $k=0.25$ to $0.0387$ at $k=0.9$, reflecting the narrowing of the resolved feature as the level approaches the peak. Table~\ref{tab:peak_levels} summarizes these quantities for the four levels shown in Figure~\ref{fig:width_example}.

Table~\ref{tab:peak_levels} also reports the half-widths $W^L_k$ and $W^R_k$, which decrease at different rates as $k$ increases, from $(W^L_{0.25},W^R_{0.25})=(0.1434,0.0696)$ at $k=0.25$ to $(W^L_{0.9},W^R_{0.9})=(0.0207,0.0179)$ at $k=0.9$. This asymmetry is captured by the index $AI_k$, also listed in the table, which decreases monotonically from $0.3462$ at $k=0.25$ to $0.2135$ at $k=0.5$, $0.1260$ at $k=0.75$, and $0.0733$ at $k=0.9$. This trend shows that the peak is markedly asymmetric near its base, decaying faster toward higher frequencies than toward lower ones, and becomes increasingly symmetric as $k$ approaches $1$, at the rate $O(\sqrt{1-k})$ established in Proposition~\ref{prop:AIclosed}.

\begin{table}[!t]
\centering
\caption{Peak, boundary frequencies, and asymmetry index for $\psi_\tau(\omega)$ with $\tau=\{0.3239\pm i0.6235\}$, at four levels $k$. The peak location $\omega_{\max}=0.1609$ and reference level $G_{\max}=0.6181$ are common to all levels, as are the minimum $G_{\min}=-0.014$ and its location $\omega_{\min}=0.3801$.}
\label{tab:peak_levels}
\begin{tabular}{lcccc}
\hline
Quantity & $k=0.25$ & $k=0.5$ & $k=0.75$ & $k=0.9$ \\
\hline
$\ell_k$         & 0.1545 & 0.3090 & 0.4636 & 0.5563 \\
$\omega^L_k$     & 0.0176 & 0.0890 & 0.1230 & 0.1402 \\
$\omega^R_k$     & 0.2306 & 0.2075 & 0.1904 & 0.1788 \\
$W^L_k$          & 0.1434 & 0.0719 & 0.0379 & 0.0207 \\
$W^R_k$          & 0.0696 & 0.0466 & 0.0294 & 0.0179 \\
$W_k$            & 0.2130 & 0.1185 & 0.0673 & 0.0387 \\
$AI_k$           & 0.3462 & 0.2135 & 0.1260 & 0.0733 \\
\hline
\end{tabular}
\end{table}

An important implication of the spectral decomposition presented in \eqref{PSD_decomp4} is that it
provides a direct interpretation of the variance of the ARMA process in terms of its
spectral components. Since the variance is given by the integral of the spectral
density over the frequency domain, we have
\begin{equation*}
\gamma_Y(0) = 2\pi \int_{-1/(2\Delta)}^{1/(2\Delta)} \PSD_Y(\omega;\vect{\beta},\sigma^2_{\eta})\,d\omega.
\end{equation*}

Substituting the decomposition of $\PSD_Y(\omega;\vect{\beta},\sigma^2_{\eta})$ given in \eqref{PSD_ARMA_with_error}, with $\PSD_X(\omega;\vect{\beta})$ given in \eqref{PSD_decomp4} and $\PSD_{\eta}(\omega;\sigma^2_{\eta})$ in \eqref{PSD_me}, into the expression for $\gamma_Y(0)$ and integrating term by term, the contribution of each class $\tau$ reduces to a simple closed form. Indeed, by \eqref{int_Ldes} and \eqref{int_S}, the variance of the process $\{Y_t\}$ can be expressed as
\begin{align}
\gamma_Y(0;\vect{\beta},\sigma^2_{\eta}) &= \mathbb{I}_{\{q\geq p\}}\, \sigma_Z^2 c_0
   +  \sum_{\substack{\tau \in \mathcal{R}/\sim \\ |\tau|=1}} 2\pi K(r_{\tau})
   +  \sum_{\substack{\tau \in \mathcal{R}/\sim \\ |\tau|=2}}
   4\pi \Real(K(r_{\tau})) + \sigma^2_{\eta}, \label{gamma0_final}
\end{align}
\noindent where $r_\tau$ denotes the representative root of class $\tau$ (namely $r_\tau = r$ when $\tau = \{r\}$ is real, and $r_\tau = r = s_\tau e^{i\alpha_\tau}$ for either choice of $r \in \{r,\bar r\}$ when $\tau = \{r,\bar r\}$ is complex).

The decomposition of the total variance in \eqref{gamma0_final} shows that it splits additively across the equivalence classes $\tau\in\mathcal{R}/\sim$. For a complex-conjugate pair, the factor of $4\pi\Real(K(r_\tau))$ arises from integrating $\psi_\tau(\omega)+\psi_\tau(-\omega)$: by \eqref{int_Ldes} and \eqref{int_S}, each $\varphi_{s_\tau}(\omega,\pm\alpha_\tau)$ term integrates to $1$, so the pair contributes $2$, weighted by $\Real(K(r_\tau))$ and the $2\pi$ normalization in \eqref{PSD_ARMA_without_error}. Each real root class instead contributes $2\pi K(r_\tau)$, coming from a single copy of this integral. When $q\geq p$, an additional term $\sigma_Z^2 c_0$ enters from the component $\PSD_{X,0}$. In contrast, the components $\zeta_{s_\tau}(\omega,\pm\alpha_\tau)$ do not contribute to the variance, they integrate to zero and only affect the shape of the spectral density, not its total mass.

Note that the quantities $K(r_\tau)$ and $\Real(K(r_\tau))$ are not necessarily positive, and can
take positive or negative values depending on the model parameters. Nevertheless, the
variance $\gamma_X(0)$ is always non-negative, as it is defined as the integral of the
spectral density, which itself is non-negative for all frequencies. Therefore, the above
expression should be interpreted as a decomposition of the variance into contributions
that may have different signs, but whose sum is guaranteed to be non-negative. For example, consider an AR$(2)$ process without measurement error. If both roots are real, forming classes
$\tau_1=\{r_{\tau_1}\}$ and $\tau_2=\{r_{\tau_2}\}$, the variance in
\eqref{gamma0_final} reduces to
\begin{equation*}
\gamma_Y(0) = 2\pi\big[K(r_{\tau_1})+K(r_{\tau_2})\big],
\end{equation*}
so positivity of $\gamma_Y(0)$ requires only that the sum be positive, not each term
individually. If instead the roots form a complex-conjugate pair
$\tau_1=\{r_{\tau_1},\bar r_{\tau_1}\}$, the variance reduces to
\begin{equation*}
\gamma_Y(0) = 4\pi\Real(K(r_{\tau_1})),
\end{equation*}
which forces $\Real(K(r_{\tau_1}))\geq 0$, so each of the two
Wrapped Cauchy components making up $\widetilde \PSD_{X,\tau_1}$ contributes a
non-negative amount to the total variance.

To further enhance the interpretability of the decomposition, we define the relative
contribution of each class as a percentage of the total variance of $\{Y_t\}$,
$\gamma_Y(0;\vect{\beta},\sigma^2_{\eta})$. Specifically, when
$q\geq p$, the constant term contributes
\begin{equation}\label{contribution_0}
\text{Contribution}_{0}(\%) = \frac{D_0}{\gamma_Y(0;\vect{\beta},\sigma^2_{\eta})} \times 100,
\qquad D_0 = \sigma_Z^2 c_0,
\end{equation}
and for each class $\tau\in\mathcal{R}/\sim$, we define
\begin{equation}\label{contribution_tau}
\text{Contribution}_{\tau}(\%) = \frac{D_{\tau}}{\gamma_Y(0;\vect{\beta},\sigma^2_{\eta})} \times 100,
\end{equation}
where $D_{\tau} = 2\pi K(r_{\tau})$ when $|\tau|=1$ and
$D_{\tau} = 4\pi\Real(K(r_{\tau}))$ when $|\tau|=2$. The contribution of the
measurement error is defined analogously as
\begin{equation}\label{contribution_eta}
\text{Contribution}_{\eta}(\%) = \frac{\sigma^2_{\eta}}{\gamma_Y(0;\vect{\beta},\sigma^2_{\eta})} \times 100.
\end{equation}
This
normalization provides a direct measure of how much each spectral component,
together with the measurement error, contributes to the total variance of the
observed process. As discussed above, the quantities
$K(r_\tau)$ and $\Real(K(r_\tau))$ are not necessarily positive,
and therefore these percentages may take negative values or exceed $100\%$ in
magnitude, even though their sum (including $D_0$ when $q\geq p$, and
$\text{Contribution}_{\eta}$) is always equal to
$100\%$. Consequently, these contributions should be interpreted as signed components
of the variance rather than independent proportions.

Table~\ref{tab:variance_contrib_arma32} shows the variance decomposition for the
ARMA$(3,2)$ model without measurement error used in Figure~\ref{fig:spectral_total}, with autoregressive
coefficients $\phi_1=0.04$, $\phi_2=-0.1$, $\phi_3=-0.3$ and moving-average
coefficients $\theta_1=0.5$, $\theta_2=0.3$. Since $q<p$, the constant term
$\PSD_{X,0}$ is absent, and the autoregressive polynomial admits one real root,
associated with class $\tau_1=\{-0.6077\}$, and one complex-conjugate pair,
associated with class $\tau_2=\{0.3239+i0.6235\}$. The real root contributes
$D_{\tau_1}=2\pi K(r_{\tau_1})=0.2361$ to the total variance, while the
complex-conjugate pair contributes $D_{\tau_2}=4\pi\,\Real\big(K(r_{\tau_2})\big)=1.2867$. The total variance
of the process is $\gamma_Y(0)=1.5228$, of which the real root accounts for
$15.5059\%$ and the complex-conjugate pair for $84.4941\%$.

\begin{table}[t]
\centering
\caption{Variance contribution of each class of the ARMA(3,2) model.}
\label{tab:variance_contrib_arma32}
\begin{tabular}{@{}ccc@{}}
\hline
        \multicolumn{1}{c}{Class} & \multicolumn{1}{c}{Variance contribution} & \multicolumn{1}{c}{Percentage (\%)} \\
        \hline
        $\{-0.6077\}$            & $0.2361$ & $15.5059$ \\
         $\{0.3239 \pm i0.6235\}$ & $1.2867$ & $84.4941$ \\
        Total & $1.5228$ & $100$ \\
        \hline
\end{tabular}
\end{table}

\section{State space representation and estimation procedure}\label{statespace_estimation}

This section describes the methodological framework for model estimation and spectral analysis. Section~\ref{state_space} presents the state space representation of ARMA models with measurement error. Section~\ref{param_estim} describes the procedure used to estimate the model parameters $\vect{\beta}$, and Section~\ref{psd_estim} presents the estimation of the spectral density from the fitted model, based on the decomposition introduced in Section~\ref{decomp_PSD_ARMA}.

\subsection{State space representation}\label{state_space}

There are several ways of expressing an ARMA model with measurement error in state-space form \citep{aoki1990statespace,box2015timeseries,brockwell1991,DEJONG2004119}. We consider the representation
\begin{eqnarray}\label{SS_ARMA_with_error}
\vect{x}_{t} &=& \matr{G} \vect{x}_{t-1} + \vect{W}_{t},\nonumber\\
y_{t} &=& \vect{F} \vect{x}_{t} + \eta_{t},
\end{eqnarray}
where $\vect{W}_{t}$ and $\eta_{t}$ are white noise processes with zero mean and covariance matrix $\matr{\Sigma}_w$ and variance $\sigma^2_{\eta}$, respectively. The state vector, transition matrix, and related quantities are given by
\begin{eqnarray}
\vect{x}_t &=&
\begin{bmatrix}
x_t^{(1)} \\
x_t^{(2)} \\
\vdots \\
x_t^{(m-1)} \\
x_t^{(m)}
\end{bmatrix},
\qquad
\matr{G} =
\begin{bmatrix}
\phi_1 &    1   & 0      & \cdots & 0      &   \\
\phi_2 &    0   & 1      & \cdots & 0      &   \\
\vdots & \vdots & \vdots & \ddots &  \vdots \\
\phi_{m-1} &  0 &  0   &   \cdots     & 1 \\
\phi_m & 0      & 0      & \cdots & 0
\end{bmatrix},
\quad
\matr{\Sigma}_w = \vect{R}\vect{R}^{\top}\sigma^2_Z,\nonumber\\
\vect{R}&=&\begin{bmatrix}
1 \\
\theta_1 \\
\vdots \\
\theta_{m-1}
\end{bmatrix},
\quad
\vect{F} = [1,0,\ldots,0,0],\label{ss_parts}
\end{eqnarray}
with $m=\max\{p,q+1\}$, $\phi_j=0$ at $j>p$, $\theta_j=0$ at $j>q$. For more details on this state-space representation see \cite{hagiwara2021time}.

In this formulation, \eqref{SS_ARMA_with_error} consists of two parts, the first defines the state equation, governing the temporal evolution of the unobserved state vector $\vect{x}_t$, while the second defines the observation equation, linking the latent process to the observed series $y_t$. In the absence of measurement error,
the observation equation simplifies to
$$y_{t} = \vect{F}\vect{x}_{t},$$
so that the observations coincide directly with the ARMA state-space representation.

\subsection{Parameter estimation}\label{param_estim}

Let $\vect{y} = (y_1, \ldots, y_n)^\top$, be a realization of the stochastic process $\{Y_{t} \}$. Assuming that the white noise terms in \eqref{ARMA_def_sin_me} and
\eqref{ARMA_def_with_me} are Gaussian, the resulting
ARMA$(p,q)$ process is also Gaussian, and the likelihood function is
\begin{align}\label{Likelihood_ARMA}
\begin{split}
p(\boldsymbol{y} \mid \vect{\beta},\sigma^2_{\eta}) &= \frac{1}{(2\pi)^{n/2} |\boldsymbol{\Sigma}|^{1/2}}
\exp \left\{ -\frac{1}{2} \boldsymbol{y}^\top
\boldsymbol{\Sigma}^{-1} \boldsymbol{y} \right\}, \\
\Sigma_{ij} &= \gamma_X(|t_{i} - t_{j}|) + \delta_{ij}\sigma^{2}_{\eta},
\end{split}
\end{align}
where
$\delta_{ij}$ is the Kronecker delta function, $\sigma^2_Z$ is the
variance of the white noise process in the ARMA model,
$\sigma^2_{\eta}$ is the variance of the measurement error for $y_t$,
and $\gamma_X(|i-j|;\vect{\beta})$ is the ACVF of the ARMA$(p,q)$ process given
by \eqref{ACVF_original}. Throughout this work, we assume
that the measurement error variance $\sigma^2_{\eta}$ is known.
Maximum likelihood estimator $\widehat{\vect{\beta}}$ of  $\vect{\beta}$ is
therefore obtained by maximizing \eqref{Likelihood_ARMA}.

Direct evaluation of \eqref{Likelihood_ARMA} requires inverting the $n \times n$ covariance
matrix $\boldsymbol{\Sigma}$, with computational complexity $O(n^3)$, which becomes computationally intensive and numerically unstable for large $n$. Expressing the ARMA$(p,q)$ process in state-space form allows the Kalman filter to evaluate the likelihood recursively, as a product of Gaussian
densities of the one-step-ahead prediction errors, operating on matrices of fixed dimension
$m \times m$ with $m = \max(p, q+1)$ and avoiding inversion of the full covariance matrix.
This reduces the overall cost to $O(n)$ \citep{Gardner1980}. Full details of the recursions
are provided in Supplementary Material S.1.

\subsection{Spectral density estimation}\label{psd_estim}

Once the parameter $\vect{\beta}$ have been estimated via maximum likelihood, the PSD of the process can be directly obtained by substituting these estimates into the corresponding analytical expression. Specifically, the plug-in estimator of the PSD is given by
\begin{equation}\label{PSD_estimated}
{\PSD}_Y(\omega;\widehat{\vect{\beta}},\sigma^2_{\eta}) ={\PSD}_X(\omega;\widehat{\vect{\beta}}) + \PSD_{\eta}(\omega,\sigma^2_{\eta}),
\end{equation}
where ${\PSD}_X(\omega; \widehat{\vect{\beta}})$ is obtained from \eqref{PSD_ARMA_without_error} by replacing the true parameters with their maximum-likelihood estimator, that is,
\begin{equation*}
\PSD_X(\omega;\widehat{\vect{\beta}})=\frac{\Delta \widehat{\sigma}^2_Z}{2\pi} \frac{ \left | \widehat{\theta}(e^{-i2\pi \omega\Delta}) \right |^2}{ \left | \widehat{\phi}(e^{-i2\pi \omega\Delta}) \right |^2 }
\end{equation*}
with
\begin{align*}
    \widehat{\phi}(z)&=1-\widehat{\phi}_1 z - \dots - \widehat{\phi}_p z^p,\\
    \widehat{\theta}(z)&=1+\widehat{\theta}_1 z +\dots + \widehat{\theta}_q z^q.
\end{align*}

Since the PSD decomposition is derived for the process $\{X_t\}$, we can likewise
construct estimators for each of the individual components of this decomposition. Specifically, let $\widehat{\mathcal{R}}$ be the set of estimated roots obtained from the maximum likelihood estimator,
we propose the plug-in estimators for all terms contributing to the PSD of $\{X_t\}$ in \eqref{PSD_decomp4}, given by
\begin{equation}
    {\PSD}_X(\omega;\widehat{\vect{\beta}}) = \mathbb{I}_{\{q\geq p\}}\, {\PSD}_{X,0}(\omega;\widehat{\vect{\beta}})
             + \sum_{\substack{\tau \in \widehat{\mathcal{R}}/\sim \\ |\tau|=1}} \widetilde{\PSD}_{X,\tau}(\omega;\widehat{\vect{\beta}})
             + \sum_{\substack{\tau \in \widehat{\mathcal{R}}/\sim \\ |\tau|=2}} \widetilde{\PSD}_{X,\tau}(\omega;\widehat{\vect{\beta}}),
             \label{PSD_decomp4_hat}
\end{equation}
where $\PSD_{X,0}(\omega;\widehat{\vect{\beta}})$ denotes the plug-in estimator of the
corresponding term in \eqref{P0}, obtained by evaluating it at the estimated coefficients
$\widehat{\vect{\beta}}$. For a class $\tau$ given by a real root (i.e., $|\tau|=1$), the term
$\widetilde{\PSD}_{X,\tau}(\omega;\widehat{\vect{\beta}})$ is obtained by plugging the
estimated coefficients $\widehat{\vect{\beta}}$ into \eqref{aux5}. Similarly, for a class
$\tau$ given by a complex-conjugate root pair (i.e., $|\tau|=2$),
$\widetilde{\PSD}_{X,\tau}(\omega;\widehat{\vect{\beta}})$ is obtained by plugging
$\widehat{\vect{\beta}}$ into \eqref{Pcomplex}. Here, $\widehat{\mathcal{R}}/\sim$ denotes the
set of equivalence classes of the estimated roots $\widehat{\mathcal{R}}$, obtained from
$\widehat{\vect{\beta}}$ in the same manner as $\mathcal{R}/\sim$ is obtained from the true
roots $\mathcal{R}$.

An alternative nonparametric estimator of the PSD is the Welch estimator \citep{Welch1967}. Let $\{y_1,\dots,y_n\}$ be a realization of a zero-mean stationary process. The data are divided into $K$ (possibly overlapping) segments of length $L$, with starting points separated by $D$ observations, such that $(K-1)D + L = n$. Denote the $k$-th segment by
\begin{equation*}
y_k(j) = y_{j + (k-1)D}, \quad j = 1, \ldots, L, \quad k = 1, \ldots, K.
\end{equation*}

For each segment, a window $w(j)$, $j=1,\ldots,L$, is applied, and the discrete
Fourier transform is computed as
\begin{equation*}
A_k(\omega)=\frac{\Delta^{1/2}}{L}\sum_{j=1}^{L}y_k(j)\, w(j)\, e^{-i 2\pi \omega \Delta j}.
\end{equation*}
The corresponding modified periodogram is
\begin{equation*}
I_k(\omega)=\,\frac{L}{U}\, \left| A_k(\omega) \right|^2,
\qquad
U = \frac{1}{L} \sum_{j=1}^{L} w^2(j).
\end{equation*}

The Welch estimator of the PSD is then the average of these modified periodograms,
\begin{equation}\label{Welch_periodogram}
\widehat{\PSD}_Y^{\,\text{W}}(\omega)
= \frac{1}{K} \sum_{k=1}^{K} I_k(\omega).
\end{equation}

The Welch estimator is asymptotically unbiased as $L\to\infty$, its expectation
converging to the true PSD, though for finite $L$ it is subject to a smoothing bias
arising from spectral leakage across neighboring frequencies. Its variance decreases as
the number of averaged segments $K$ increases, reflecting the classical bias-variance
trade-off of nonparametric spectral estimation. Mean-square consistency therefore
requires both $L$ and $K$ to increase with the sample size $n$.

In contrast to nonparametric methods such as the Welch estimator, the plug-in estimator
defined in \eqref{PSD_estimated} exploits the ARMA model structure, yielding lower
variance and improved interpretability, particularly in the presence of spectral peaks
\citep{kay1981}.

Beyond the PSD itself, the plug-in principle also provides estimators of the quantities used to characterize each spectral component. Given the estimated component $\psi_{\tau}(\omega;\widehat{\vect{\beta}})$, the peak frequency $\omega_{\max}$, the boundaries $\omega_{k}^{L}$, $\omega_{k}^{R}$, and the associated widths $W_{k}^{L}$, $W_{k}^{R}$, and $W_{k}$ are obtained by applying their corresponding definitions to $\psi_{\tau}(\omega;\widehat{\vect{\beta}})$ in place of $\psi_{\tau}(\omega;\vect{\beta})$, yielding the plug-in estimators $\widehat\omega_{\max}$, $\widehat\omega_{k}^{L}$, $\widehat\omega_{k}^{R}$, $\widehat W_{k}^{L}$, $\widehat W_{k}^{R}$, and $\widehat W_{k}$.

The same principle applies to the contribution of each component to the total variance of the process. Replacing $\vect{\beta}$ with $\widehat{\vect{\beta}}$ in the variance decomposition \eqref{gamma0_final}, while keeping the known measurement error variance $\sigma^2_{\eta}$ unchanged, yields the plug-in estimator of the total variance,
\begin{align*}
\gamma_Y(0;\widehat{\vect{\beta}},\sigma^2_{\eta}) = \mathbb{I}_{\{q\geq p\}}\, \widehat\sigma_Z^2 \widehat c_0
   +  \sum_{\substack{\tau \in \widehat{\mathcal{R}}/\sim \\ |\tau|=1}} 2\pi K( r_{\tau})
   +  \sum_{\substack{\tau \in \widehat{\mathcal{R}}/\sim \\ |\tau|=2}}
   4\pi \Real(K(r_{\tau})) + \sigma^2_{\eta}.
\end{align*}

Analogously, the estimated contribution of each term to the total variance is
obtained via the plug-in principle, replacing $D_0$ and $D_\tau$ in the
definitions \eqref{contribution_0} and \eqref{contribution_tau} by their plug-in
counterparts $\widehat D_0$ and $\widehat D_\tau$, and evaluating the denominator
at $\gamma_Y(0;\widehat{\vect{\beta}},\sigma^2_{\eta})$. The contribution of the
measurement-error component retains the form \eqref{contribution_eta}, evaluated
at the same denominator.

\section{Numerical study}\label{simulations}

Propositions~\ref{prop:peak} to \ref{prop:AIclosed} give the location, width and
asymmetry of a spectral peak in closed form as functions of the ARMA parameters.
This section assesses whether the plug-in versions of those closed forms recover
the true descriptors, and how they compare with the two quantities, namely the nominal centroid $\alpha/(2\pi\Delta)$ and a numerical
maximization of the fitted spectral density.

We consider three models, specified through the equivalence classes
$\tau\in\mathcal{R}/\!\sim$ of Section~\ref{decomp_PSD_ARMA} and chosen so that
each exercises a different aspect of the decomposition.

\textbf{M1: AR($2$), one sharply concentrated class.} A single complex class
$\tau_1=\{0.9\,e^{\pm i\,0.4\pi}\}$, so that $\mathcal{R}/\!\sim\,=\{\tau_1\}$
with $n_R=0$ and $n_C=1$, giving $\vect\phi=(0.5562,-0.81)$ and
$\gamma_X(0)=3.2111$. This is the regime that dominates the astronomical
literature, in which the component is narrow.

\textbf{M2: AR($4$), one sharp and one broad class.} Two complex classes,
$\tau_1=\{0.95\,e^{\pm i\,0.2\pi}\}$ and $\tau_2=\{0.70\,e^{\pm i\,0.6\pi}\}$, so
$n_R=0$ and $n_C=2$, giving $\vect\phi=(1.1045,-0.7275,0.3628,-0.4422)$ and
$\gamma_X(0)=5.8496$. The class $\tau_2$ is weakly concentrated, so by
Proposition~\ref{prop:peak} its peak is displaced appreciably from
$\alpha_2/(2\pi\Delta)$, while $\tau_1$ is sharply concentrated; the model also
allows the resolution of two distinct components to be assessed.

\textbf{M3: ARMA($3,2$), one complex and one real class.} A complex class
$\tau_1=\{0.9\,e^{\pm i\,0.3\pi}\}$ and a real class $\tau_2=\{-0.6\}$, so
$n_R=1$ and $n_C=1$, with $\theta_1=0.5$ and $\theta_2=0.3$, giving
$\vect\phi=(0.458,-0.1752,-0.486)$ and $\gamma_X(0)=4.2189$. Since $q>0$, the
weights $K(r)$ depend on the moving-average polynomial through
$\theta(r)\theta(r^{-1})$ and acquire an imaginary part five times larger than in
the pure autoregressive models, so this is the model in which the asymmetry index
of Definition~\ref{def:asym} is informative. The real class contributes a
broadband background, and by \eqref{aux5} it has no interior peak, so the
descriptors are reported for the complex class only.

Table~\ref{tab:sim_models} reports the true descriptors, and
Figure~\ref{fig:sim_psd} shows the corresponding spectral densities. The
displacement $\omega_{\max}-\alpha/(2\pi\Delta)$ is at most of order $10^{-3}$ for the three sharply concentrated classes and of order $10^{-2}$ for $\tau_2$ of M2, a difference of two orders of magnitude; the asymmetry index behaves likewise.

\begin{table}[b]
\centering
\caption{True values of the peak descriptors for each complex class
$\tau\in\mathcal{R}/\!\sim$, at level $k=1/2$ and $\Delta=1$. Here $s$ and
$\alpha$ are the polar coordinates of the representative root of $\tau$, $A$ and
$B$ are the real and imaginary parts of $K(r)$ in \eqref{Kr}, and $\omega_{\max}$,
$W_{1/2}$ and $AI_{1/2}$ follow from Propositions~\ref{prop:peak} and
\ref{prop:length}, and Definition \ref{def:asym}.}
\label{tab:sim_models}
\begin{tabular}{llrrrrrrrr}
\hline
Model & Class & $s$ & $\alpha$ & $A$ & $B$ & $\alpha/(2\pi\Delta)$ &
$\omega_{\max}$ & $W_{1/2}$ & $AI_{1/2}$ \\
\hline
M1 & $\tau_1$ & 0.9000 & 1.2566 & 0.2555 & -0.0087 & 0.2000 & 0.1997 & 0.0336 & 0.0170 \\
M2 & $\tau_1$ & 0.9500 & 0.6283 & 0.4456 & -0.0098 & 0.1000 & 0.0999 & 0.0163 & 0.0110 \\
M2 & $\tau_2$ & 0.7000 & 1.8850 & 0.0199 & -0.0104 & 0.3000 & 0.2863 & 0.1203 & 0.2247 \\
M3 & $\tau_1$ & 0.9000 & 0.9425 & 0.3253 & -0.0482 & 0.1500 & 0.1488 & 0.0337 & 0.0731 \\
\hline
\end{tabular}
\end{table}

\begin{figure}[!t]
\centering
\includegraphics[width=14cm]{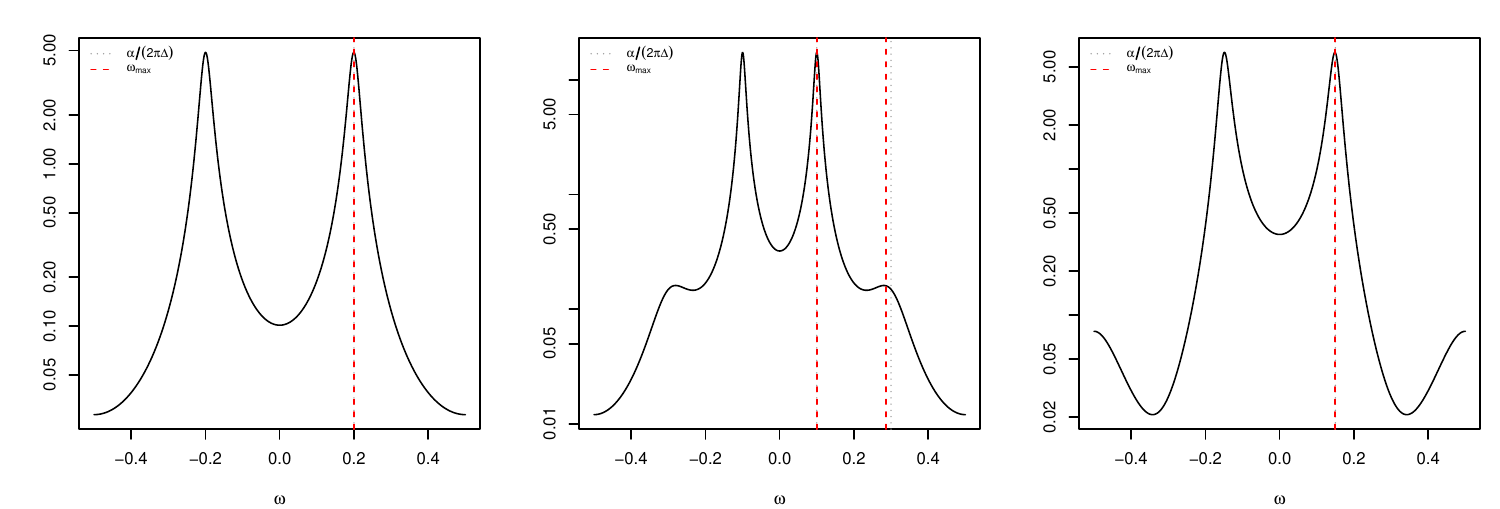}
\caption{Spectral densities of AR($2$) (left panel), AR($4$) (center panel), and ARMA($3,2$) (right panel), with
$\sigma^2_Z=1$ and $\Delta=1$. For each complex class the dotted line marks the
nominal centroid $\alpha/(2\pi\Delta)$ and the dashed line the peak
$\omega_{\max}$ of Proposition~\ref{prop:peak}. For the sharply concentrated
classes the two are indistinguishable at this resolution; they separate visibly
only for the broad class $\tau_2$ of M2.}
\label{fig:sim_psd}
\end{figure}

For each model we simulate $R=1000$ series of length $n\in\{100,1000\}$ at
$\Delta=1$, with $\sigma^2_Z=1$. Measurement error is added at two levels of the
noise-to-signal ratio (nsr) $\sigma^2_\eta/\gamma_Y(0)\in\{0,0.20\}$, the second close
to the value found in the application of Section~\ref{application}. Throughout,
$\sigma^2_\eta$ is treated as known, as in Section~\ref{param_estim}.

Each series is fitted by maximum likelihood through the state-space
representation and the Kalman filter of Section~\ref{statespace_estimation}, over
$p\in\{1,\dots,5\}$ and $q\in\{0,\dots,5\}$, with the order selected by Bayesian Information Criterion (BIC) \citep{Schwarz1978};
results for Akaike Information Criterion (AIC) \citep{Akaike1974} are also recorded.

From each fitted model we take the complex classes of
$\widehat{\mathcal{R}}/\!\sim$ and match each to a true class by nearest
argument, recording a failure when no estimated class lies within $0.05$ of the
true argument. For each matched class we compute three estimators of the same
target $\omega_{\max}$: $\widehat\omega_{\max}$, the closed form \eqref{eq:wmax}
evaluated at $\widehat{\vect\beta}$;
$\widehat\omega_{\mathrm{cen}}=\widehat\alpha/(2\pi\Delta)$, the nominal
centroid; and $\widehat\omega_{\mathrm{grid}}$, the local maximum of
$\widehat\PSD_X$ nearest the true peak, found numerically on a grid of $2001$
points over $(0,1/(2\Delta)]$. We also compute $\widehat W_{1/2}$ from \eqref{eq:width} and $\widehat{AI}_{1/2}$ from \eqref{eq:AIclosed}.

Table~\ref{tab:sim_freq} reports bias and root mean squared error (RMSE) for the three
estimators of $\omega_{\max}$. For the three sharply concentrated classes the
estimators are indistinguishable: at $n=1000$ without measurement error the RMSE
of all three lies between $0.0012$ and $0.0027$, and the bias of the centroid,
$0.0004$ for M1 and $0.0001$ for $\tau_1$ of M2, agrees with the displacements
of $0.0003$ and $0.0001$ in Table~\ref{tab:sim_models}. In this regime the
nominal centroid is an adequate estimator of $\omega_{\max}$, as it is usually
taken to be.

For the broad class $\tau_2$ of M2 the centroid carries a bias of $0.0137$ at
$n=1000$, against a displacement of $0.0137$ predicted by
Proposition~\ref{prop:peak}, and this bias does not diminish with $n$: it is
$0.0141$ at $n=100$ and $0.0137$ at $n=1000$, while the closed form has biases
of $0.0011$ and $-0.0001$ respectively. At $n=1000$ the centroid bias is three
and a half times the standard deviation of the closed-form estimator, and its
RMSE is three and a half times larger. The same agreement holds across the
design: in all sixteen cells, spanning four classes, two sample sizes and two
noise levels, the empirical bias of the centroid differs from the predicted
displacement by at most $0.0021$, against a displacement of $0.0137$ for
$\tau_2$ of M2.

The numerical maximization fails for that same class, and for a different reason.
Its RMSE is $0.0322$ at $n=1000$ and $0.0816$ at $n=100$, an order of magnitude
worse than either analytical estimator, with a standard deviation of $0.0301$.
The maximum of the total spectral density is not the maximum of the class
component, and when the component is broad the two are far apart and the
local-maximum detection becomes unstable. This is precisely the situation in which
a closed form for the component is needed, since no search on $\widehat\PSD_X$
recovers $\omega_{\max}$ for the class.

For M3 the centroid bias of $0.0016$ at $n=1000$ decomposes as the closed-form
bias $0.0004$, which is finite-sample estimation error, plus the predicted
displacement $0.0012$.

\begin{table}[t]
\centering
\caption{Bias and RMSE of the three estimators of $\omega_{\max}$, over $1000$
replications, with the order selected by BIC.}
\label{tab:sim_freq}
\begin{tabular}{llrrrrrrr}
\hline
& & & \multicolumn{2}{c}{$\widehat\omega_{\max}$}
& \multicolumn{2}{c}{$\widehat\omega_{\mathrm{cen}}$}
& \multicolumn{2}{c}{$\widehat\omega_{\mathrm{grid}}$} \\
Class & $n$ & nsr & Bias & RMSE & Bias & RMSE & Bias & RMSE \\
\hline
M1, $\tau_1$ & 100  & 0    & -0.0003 & 0.0041 & 0.0001 & 0.0040 & -0.0003 & 0.0041 \\
M2, $\tau_1$ & 100  & 0    & -0.0003 & 0.0034 & -0.0002 & 0.0035 & -0.0003 & 0.0034 \\
M2, $\tau_2$ & 100  & 0    & 0.0011 & 0.0079 & 0.0141 & 0.0147 & -0.0360 & 0.0814 \\
M3, $\tau_1$ & 100  & 0    & 0.0000 & 0.0043 & 0.0010 & 0.0044 & 0.0000 & 0.0043 \\
\hline
M1, $\tau_1$ & 1000 & 0    & 0.0001 & 0.0016 & 0.0004 & 0.0017 & 0.0001 & 0.0016 \\
M2, $\tau_1$ & 1000 & 0    & -0.0000 & 0.0012 & 0.0001 & 0.0012 & -0.0000 & 0.0012 \\
M2, $\tau_2$ & 1000 & 0    & -0.0001 & 0.0040 & 0.0136 & 0.0140 & -0.0114 & 0.0322 \\
M3, $\tau_1$ & 1000 & 0    & 0.0004 & 0.0022 & 0.0016 & 0.0027 & 0.0004 & 0.0022 \\
\hline
M1, $\tau_1$ & 100  & 0.20 & 0.0000 & 0.0042 & 0.0004 & 0.0042 & 0.0000 & 0.0042 \\
M2, $\tau_1$ & 100  & 0.20 & 0.0005 & 0.0037 & 0.0007 & 0.0038 & 0.0004 & 0.0037 \\
M2, $\tau_2$ & 100  & 0.20 & 0.0137 & 0.0142 & 0.0155 & 0.0160 & 0.0136 & 0.0142 \\
M3, $\tau_1$ & 100  & 0.20 & -0.0008 & 0.0041 & 0.0000 & 0.0039 & -0.0008 & 0.0041 \\
\hline
M1, $\tau_1$ & 1000 & 0.20 & 0.0000 & 0.0020 & 0.0003 & 0.0020 & 0.0000 & 0.0020 \\
M2, $\tau_1$ & 1000 & 0.20 & -0.0004 & 0.0016 & -0.0003 & 0.0017 & -0.0004 & 0.0016 \\
M2, $\tau_2$ & 1000 & 0.20 & 0.0030 & 0.0058 & 0.0136 & 0.0142 & -0.0020 & 0.0182 \\
M3, $\tau_1$ & 1000 & 0.20 & -0.0018 & 0.0034 & -0.0009 & 0.0032 & -0.0018 & 0.0034 \\
\hline
\end{tabular}
\end{table}

Figure~\ref{fig:sim_box} shows the corresponding sampling distributions. For the
three sharply concentrated classes the three estimators overlap almost exactly,
over a range of at most $0.008$; for $\tau_2$ of M2 the centroid is displaced
from the true value by more than the whole interquartile range of the closed
form, and the numerical maximization is both biased and far more dispersed.

\begin{figure}[b]
\centering
\includegraphics[width=14cm]{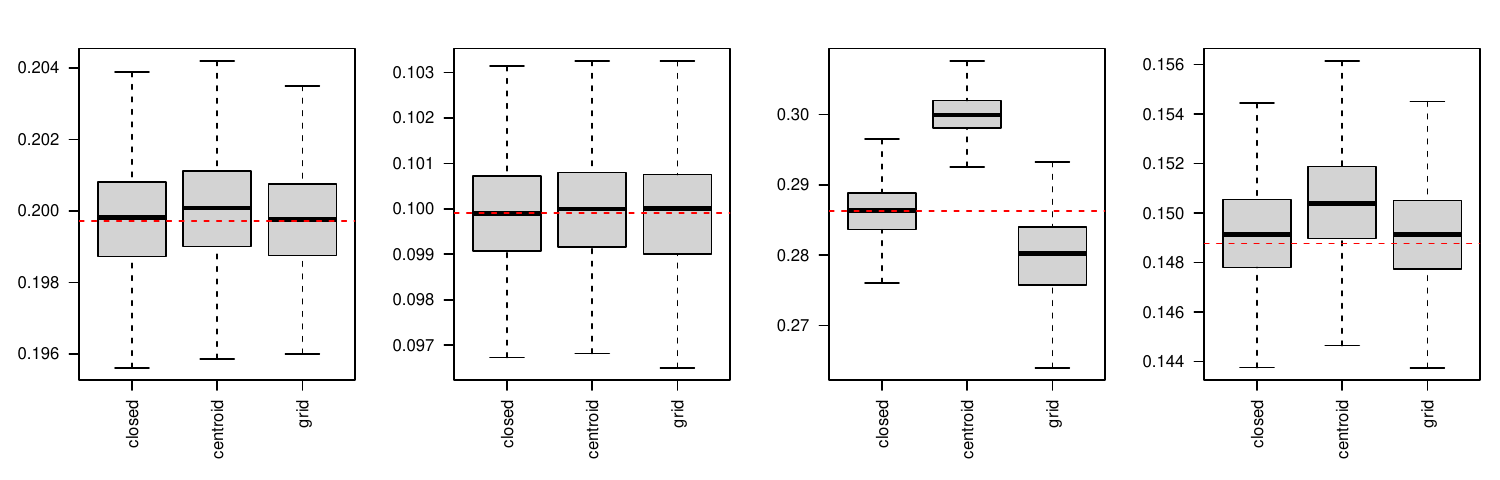}
\caption{Sampling distributions of the three estimators of $\omega_{\max}$ over
$1000$ replications at $n=1000$ without measurement error, with the order
selected by BIC. The dashed line marks the true $\omega_{\max}$ of
Table~\ref{tab:sim_models}. Note the differing vertical scales.}
\label{fig:sim_box}
\end{figure}

We turn to the width and the asymmetry. Table~\ref{tab:sim_shape} reports the
same summaries for $W_{1/2}$ and $AI_{1/2}$. Neither has a counterpart among the
quantities ordinarily read off a fitted root, so the closed forms are the only
estimators available.

At $n=1000$ without measurement error both are recovered with small bias. The
width has a relative bias of at most $2.7\%$ of its true value and a relative
standard deviation (SD) between $8.8\%$ and $16.6\%$. The asymmetry index, which by
Proposition~\ref{prop:AIclosed} is a ratio of two quantities that both vanish as
the level approaches the maximum, is nonetheless estimable: its relative standard
deviation is $7.3\%$ for $\tau_2$ of M2, where $AI_{1/2}=0.2247$, and $17.9\%$
for M3, where $AI_{1/2}=0.0731$, against $13.2\%$ for M1, where the index is an
order of magnitude smaller.

At $n=100$ the width bias reaches $15\%$ of its true value for M3 without
measurement error, and $69\%$ for $\tau_2$ of M2 with it, so the descriptors
should not be read at that sample size.

\begin{table}[t]
\centering
\caption{Estimation of the width and the asymmetry index at $k=1/2$, over $1000$
replications, with the order selected by BIC.}
\label{tab:sim_shape}
\begin{tabular}{llrrrrrr}
\hline
& & & \multicolumn{2}{c}{$W_{1/2}$} & \multicolumn{2}{c}{$AI_{1/2}$} \\
\multicolumn{1}{c}{Class} & \multicolumn{1}{c}{$n$} & \multicolumn{1}{c}{nsr} & \multicolumn{1}{c}{Bias} & \multicolumn{1}{c}{SD} & \multicolumn{1}{c}{Bias} & \multicolumn{1}{c}{SD} \\
\hline
M1, $\tau_1$ & 100  & 0    & 0.0023 & 0.0127 & 0.0011 & 0.0085 \\
M2, $\tau_1$ & 100  & 0    & 0.0019 & 0.0090 & -0.0001 & 0.0143 \\
M2, $\tau_2$ & 100  & 0    & -0.0056 & 0.0293 & -0.0126 & 0.0462 \\
M3, $\tau_1$ & 100  & 0    & -0.0051 & 0.0115 & -0.0162 & 0.0280 \\
\hline
M1, $\tau_1$ & 100  & 0.20 & 0.0019 & 0.0157 & 0.0006 & 0.0124 \\
M2, $\tau_1$ & 100  & 0.20 & 0.0027 & 0.0119 & 0.0107 & 0.0395 \\
M2, $\tau_2$ & 100  & 0.20 & -0.0831 & 0.0193 & -0.1482 & 0.0362 \\
M3, $\tau_1$ & 100  & 0.20 & 0.0003 & 0.0167 & -0.0322 & 0.0193 \\
\hline
M1, $\tau_1$ & 1000 & 0    & 0.0003 & 0.0039 & 0.0002 & 0.0022 \\
M2, $\tau_1$ & 1000 & 0    & 0.0004 & 0.0026 & 0.0003 & 0.0021 \\
M2, $\tau_2$ & 1000 & 0    & -0.0001 & 0.0106 & -0.0005 & 0.0165 \\
M3, $\tau_1$ & 1000 & 0    & -0.0009 & 0.0056 & 0.0000 & 0.0131 \\
\hline
M1, $\tau_1$ & 1000 & 0.20 & 0.0001 & 0.0045 & 0.0000 & 0.0031 \\
M2, $\tau_1$ & 1000 & 0.20 & 0.0012 & 0.0035 & -0.0061 & 0.0118 \\
M2, $\tau_2$ & 1000 & 0.20 & -0.0166 & 0.0157 & -0.0254 & 0.0243 \\
M3, $\tau_1$ & 1000 & 0.20 & -0.0005 & 0.0052 & -0.0237 & 0.0183 \\
\hline
\end{tabular}
\end{table}

The results above are conditional on the true class being identified; we now
report how often that happens. Table~\ref{tab:sim_sel} gives the proportion of
replications in which the true order was selected and in which each true class
was matched. Without measurement error BIC recovers the order in $99.5\%$ of
replications for M1 and $99.4\%$ for M2 at $n=1000$, but in only $76.0\%$ for M3,
and AIC is markedly worse throughout, between $41\%$ and $65\%$. The difficulty
for M3 is the moving-average part: the ARMA$(3,2)$ likelihood is nearly flat in
directions that trade $\theta$ against $\phi$, and BIC's penalty favours the more
parsimonious autoregressive alternatives. With measurement error this becomes
acute, the true order being selected in $1.5\%$ of replications at $n=1000$,
although the complex class is still matched in $97.6\%$ of them.

Matching is essentially perfect at $n=1000$ without measurement error, at $1.000$,
$0.988$ and $0.998$ for the three models. Measurement error affects $\tau_2$ of
M2 most: the match rate falls to $0.622$, and at $n=100$ to $0.408$, at which
point the closed-form estimator of $\omega_{\max}$ acquires a bias of $0.0137$,
indistinguishable from the displacement itself, and the width and asymmetry
estimates break down entirely, with biases of $-0.083$ and $-0.148$. The flat
term \eqref{PSD_me} raises the spectral floor without altering any descriptor,
but it reduces the prominence of a weakly concentrated component to the point
where the corresponding root is no longer identified.

\begin{table}[t]
\centering
\caption{Proportion of replications in which the true order was selected, and in
which each true class was matched, over $1000$ replications.}
\label{tab:sim_sel}
\begin{tabular}{llrrrrrr}
\hline
& & \multicolumn{2}{c}{selection, BIC}
& \multicolumn{2}{c}{selection, AIC}
& \multicolumn{2}{c}{matched, BIC} \\
\multicolumn{1}{c}{Model}
& \multicolumn{1}{c}{nsr}
& \multicolumn{1}{c}{$n=100$}
& \multicolumn{1}{c}{$n=1000$}
& \multicolumn{1}{c}{$n=100$}
& \multicolumn{1}{c}{$n=1000$}
& \multicolumn{1}{c}{$n=100$}
& \multicolumn{1}{c}{$n=1000$} \\
\hline
M1 & 0    & 0.950 & 0.995 & 0.523 & 0.412 & 0.820 & 1.000 \\
M2 & 0    & 0.829 & 0.994 & 0.528 & 0.654 & 0.672 & 0.988 \\
M3 & 0    & 0.083 & 0.760 & 0.162 & 0.565 & 0.671 & 0.998 \\
M1 & 0.20 & 0.938 & 0.992 & 0.612 & 0.642 & 0.707 & 0.999 \\
M2 & 0.20 & 0.072 & 0.592 & 0.218 & 0.644 & 0.408 & 0.622 \\
M3 & 0.20 & 0.001 & 0.015 & 0.045 & 0.194 & 0.658 & 0.976 \\
\hline
\end{tabular}
\end{table}

The estimators considered so far are all derived from a fitted ARMA model. We now
compare them with the two-step procedure they are intended to replace: a Welch
estimate of the spectral density followed by a least-squares fit of a sum of
Lorentzians. For each simulated series the Welch estimator
\eqref{Welch_periodogram} is computed with Hann-windowed segments overlapping by
one half, and a sum of $n_C$ Lorentzians plus a constant floor is fitted to it on
the logarithmic scale, with starting centroids at the highest local maxima of the
estimate. Two choices favour the competitor: the number of Lorentzians is fixed
at the true number of complex classes, so it is given the model order, whereas
the ARMA route must select $p$ and $q$ by BIC; and a fitted component is matched
to a true class when its centroid lies within $0.02$ of $\omega_{\max}$, a more
generous tolerance than the one used for the ARMA classes, since the Welch
resolution is $1/L$. Both routes are applied to the same series, so the
comparison is paired. Results are reported at $n=1000$ without measurement error
for segment lengths $L=124$ and $L=250$, averaging $K=15$ and $K=7$ segments
respectively; at $n=100$ the resolution would be coarser than the width of the
sharply concentrated classes, so the comparison there would measure the
resolution rather than the method.

Table~\ref{tab:sim_comp} reports the outcome. The closed form of
Proposition~\ref{prop:peak} estimates the peak frequency with a root mean squared
error between $1.7$ and $6.9$ times smaller than the Lorentzian centroid, and the
ordering is unchanged when the segment length doubles. On the same series the
closed form is closer to the true $\omega_{\max}$ than the centroid in $90.5\%$
of replications for M1, $99.8\%$ for M3, and $70.1\%$ and $73.4\%$ for the two
classes of M2. The margin is largest for M3, where the centroid carries a bias of
$-0.0137$, an order of magnitude larger than in any other configuration: the
Lorentzian model has no term for the broadband contribution of the real class, so
that contribution is absorbed into the constant floor and displaces the fitted
centroid. Since broadband structure is the rule rather than the exception in
applications, this is the configuration most representative of practice.

The comparison of the widths is subject to the caveat that the two quantities
target slightly different objects, the Lorentzian width being that of a function
fitted to $\PSD_Y$ while $W_{1/2}$ is that of the class component of $\PSD_X$.
Even so, the gap is larger than for the peak frequency, with ratios between $1.1$
and $6.6$. For the broad class $\tau_2$ of M2 the fitted full width at half
maximum underestimates $W_{1/2}$ by $0.057$, roughly half its true value of
$0.120$, the Lorentzian shape being unable to reproduce a broad component whose
wings are partly absorbed by the floor. One fitted width out of $377$ at $L=250$
diverged and was discarded.

A further difference does not appear in the error columns, which are conditional
on a component being found at all. The ARMA route matched the true class in
$97.8\%$ of replications or better in every configuration, whereas the two-step
route matched $\tau_2$ of M2 in only $57.1\%$ of replications at $L=124$ and
$37.7\%$ at $L=250$, and M1 in about $84\%$. Longer segments give finer
resolution, from $0.0081$ to $0.0040$, but halve the number averaged, from $15$
to $7$, so the estimate is noisier. For a component whose width is $0.120$, thirty
times the coarser resolution, the binding constraint is variance rather than
resolution, and the match rate falls accordingly. The errors reported for that
class are therefore computed over the minority of replications in which the
two-step route located it.

\begin{table}[t]
\centering
\caption{Comparison with the two-step procedure at $n=1000$ without measurement
error, over $1000$ replications, for two Welch segment lengths. The closed form
of Proposition~\ref{prop:peak} and the width \eqref{eq:width} are computed from
the ARMA fit with the order selected by BIC; the centroid and full width at half
maximum are obtained by fitting $n_C$ Lorentzians to the Welch estimate. Error
summaries are conditional on the class being matched.}
\label{tab:sim_comp}
\begin{tabular}{llrrrrrrrr}
\hline
& & \multicolumn{2}{c}{matched} & \multicolumn{3}{c}{$\omega_{\max}$}
  & \multicolumn{3}{c}{$W_{1/2}$} \\
\multicolumn{1}{c}{Class} & \multicolumn{1}{c}{$L$} & \multicolumn{1}{c}{ARMA} & \multicolumn{1}{c}{Welch} & \multicolumn{1}{c}{RMSE$_{\mathrm{A}}$} & \multicolumn{1}{c}{RMSE$_{\mathrm{W}}$} & \multicolumn{1}{c}{ratio}
  & \multicolumn{1}{c}{RMSE$_{\mathrm{A}}$} & \multicolumn{1}{c}{RMSE$_{\mathrm{W}}$} & \multicolumn{1}{c}{ratio} \\
\hline
M1, $\tau_1$ & 124 & 1.000 & 0.839 & 0.00171 & 0.00502 & 2.94 & 0.00370 & 0.01217 & 3.29 \\
M2, $\tau_1$ & 124 & 1.000 & 1.000 & 0.00116 & 0.00200 & 1.72 & 0.00269 & 0.01182 & 4.39 \\
M2, $\tau_2$ & 124 & 0.978 & 0.571 & 0.00439 & 0.00896 & 2.04 & 0.01085 & 0.05895 & 5.43 \\
M3, $\tau_1$ & 124 & 1.000 & 0.984 & 0.00202 & 0.01395 & 6.91 & 0.00537 & 0.00597 & 1.11 \\
\hline
M1, $\tau_1$ & 250 & 1.000 & 0.834 & 0.00171 & 0.00510 & 2.99 & 0.00370 & 0.01141 & 3.09 \\
M2, $\tau_1$ & 250 & 1.000 & 1.000 & 0.00116 & 0.00201 & 1.73 & 0.00269 & 0.01773 & 6.59 \\
M2, $\tau_2$ & 250 & 0.978 & 0.377 & 0.00439 & 0.00944 & 2.15 & 0.01085 & 0.06950 & 6.41 \\
M3, $\tau_1$ & 250 & 1.000 & 0.970 & 0.00202 & 0.01313 & 6.51 & 0.00537 & 0.01415 & 2.63 \\
\hline
\end{tabular}
\end{table}

Taken together, the results above support four conclusions. The closed form of
Proposition~\ref{prop:peak} recovers the peak location without bias in every
configuration in which the underlying root is identified. The nominal centroid is
adequate when the component is sharply concentrated, which is the usual situation
in the astronomical literature, but when it is not, it carries a bias equal to
the displacement the proposition predicts: for $\tau_2$ of M2 the empirical bias
is $0.01365$ at $n=1000$, against a predicted $0.01372$, and it does not diminish
with the sample size. A numerical search on the fitted spectral density is not a
substitute, since for a broad component the maximum of the total density is not
the maximum of the class component, and its RMSE exceeds that of the closed form
by an order of magnitude. Finally, the width and the asymmetry index, which have no counterpart among the
quantities read off a fitted root, are recovered with small bias and useful
precision at $n=1000$ whenever the class is identified.

Lastly, both closed forms outperform the two-step procedure of estimating the
spectrum nonparametrically and fitting Lorentzians to it, by a factor between
$1.7$ and $6.9$ in root mean squared error for the peak frequency and by more for
the width, even when that procedure is given the correct number of components and
a more generous matching tolerance.

\section{Application}\label{application}

In this section, we apply the PSD decomposition for ARMA with measurement error
models introduced in Section \ref{decomp_PSD_ARMA} by fitting these models
to astronomical time series data. Specifically, we study the PSD of a light
curve from the X-ray binary XTE 1550-564, observed by the Rossi X-ray Timing
Explorer (RXTE). This light curve was previously analyzed by \cite{Kelly2014}
and was selected because it is regularly sampled at $\Delta t = 1/128$ seconds
and exhibits a complex, well-characterized PSD.

We analyze a segment of approximately $7.8047$ seconds, containing $n = 1000$
observations, from $7.812$ seconds to $15.617$ seconds. Since X-ray binary light
curves are known to follow a lognormal distribution \citep{Uttley2005}, we
apply a logarithmic transformation to obtain data that are Gaussian. It is well established that this source is affected by a measurement error
\citep{Kelly2014}. Therefore, its variance was estimated following \cite{Vaughan2003}, yielding
$\sigma^2_\eta = 0.0107$, which is subsequently incorporated into the modeling framework.

The transformed light curve and its Welch periodogram are shown in
Figure~\ref{Figure:DATA_APPLI}. The Welch estimator in \eqref{Welch_periodogram}
was computed using a Hann window with segments of length $L = 500$ and
50\% overlap ($D = 250$), yielding $K = 3$ segments for the $n = 1000$
observations. The periodogram reveals a dominant
QPO at approximately $2.816$ Hz, consistent
with the strong QPO reported by \cite{Kelly2014} for this source. A second,
weaker QPO is also reported by \cite{Kelly2014} at approximately $6$ Hz;
in our segment this feature appears at $6.144$ Hz, though it is less
prominent, possibly due to the shorter length of the analyzed segment.
The periodogram also shows a flattening at high frequencies attributable
to measurement noise.

\begin{figure}[!b]
\centering
\includegraphics[width=13.5cm]{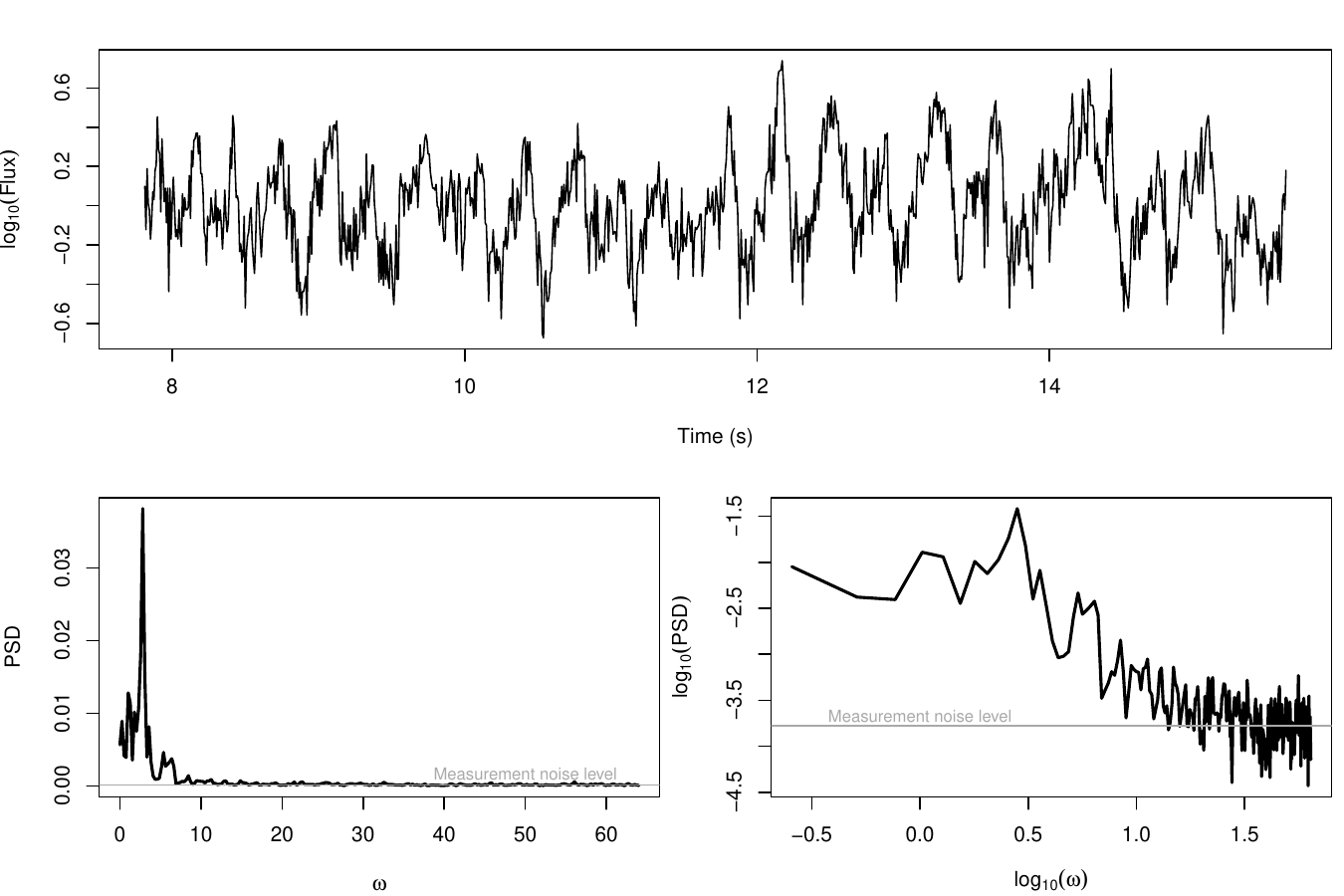}
\caption{RXTE light curve of the X-ray binary XTE~1550-564, comprising
$n=1000$ observations (upper panel). The Welch periodogram defined in \eqref{Welch_periodogram} of the light
curve is shown in linear scale and in log-log scale in the lower panels,
where the horizontal grey line indicates the measurement noise level.}
\label{Figure:DATA_APPLI}
\end{figure}

We fit ARMA$(p,q)$ models with measurement error to the logarithmic light curve of XTE J1550-564, estimating the parameters $\phi_1,\ldots,\phi_p$, $\theta_1,\ldots,\theta_q$, and $\sigma_Z^2$ ($\sigma_\eta^2$ is known) via a state space representation combined with the Kalman filter (see Section~\ref{statespace_estimation}), which enables likelihood-based inference in the presence of measurement error. We consider models with orders $p,q \in \{0,1,\ldots,6\}$, allowing additional flexibility to accommodate a broadband noise component or other spectral structure if needed.

To evaluate the best fitting model, we use AIC and BIC. Under the AIC, the three models with the lowest values are the ARMA$(2,1)$, ARMA$(3,1)$, and ARMA$(3,2)$ specifications, with AIC values of $-988.83$, $-986.91$, and $-984.68$, respectively. Under the BIC, the three models with the lowest values are the ARMA$(2,1)$, ARMA$(3,1)$, and ARMA$(2,0)$, with BIC values of $-964.30$, $-957.47$, and $-955.20$, respectively. However, inspection of the estimated spectral density reveals that neither the top AIC- nor BIC-ranked models succeed in reproducing the dominant QPO at approximately $2.816$~Hz. The first model that successfully captures this spectral feature is the ARMA$(3,2)$ specification, which ranks third under the AIC and seventh under the BIC, while remaining competitive in terms of both criteria. The ARMA$(3,2)$ model is therefore selected as the preferred specification, as it provides the most parsimonious representation that adequately captures the dominant QPO of the source.

Table~\ref{tab:summary_fit_ARMA} summarizes the parameter estimates of the fitted ARMA$(3,2)$ model, along with their standard deviations and $95\%$ confidence intervals (CI) computed via the parametric bootstrap procedure described in Supplementary Material S.2, based on $B = 1000$ replications. Notably, the $95\%$ CI for all parameters $\phi_1, \phi_2, \phi_3,
\theta_1$, and $\theta_2$ exclude zero, providing evidence that all estimates are statistically significant at the $5\%$ level and offering empirical support for the selected ARMA$(3,2)$ specification.

\begin{table}[b]
\centering
\caption{Parameters estimated fitting an ARMA($3,2$) model to the RXTE light curve of the X-ray binary XTE 1550-564.}
\label{tab:summary_fit_ARMA}
\begin{tabular}{@{}crrc@{}}
\hline
\multicolumn{1}{c}{Parameter} & \multicolumn{1}{c}{Estimate}
& \multicolumn{1}{c}{SD} & \multicolumn{1}{c}{CI} \\
\hline
$\phi_1$      &  2.8535 & 0.0156 & $[2.8188, 2.8915]$ \\
$\phi_2$      & -2.7304 & 0.0302 & $[-2.8043,-2.6636]$ \\
$\phi_3$      &  0.8748 & 0.0149 & $[0.8424, 0.9111]$ \\
$\theta_1$    & -1.9303 & 0.0225 & $[-1.9785, -1.8855]$ \\
$\theta_2$    &  0.9486 & 0.0230 & $[0.8999, 0.9971]$ \\
$\sigma^2_Z$  &  0.0077 & ---    & --- \\ \hline
\end{tabular}
\end{table}

\begin{figure}[htb!]
\centering
\includegraphics[width=14cm]{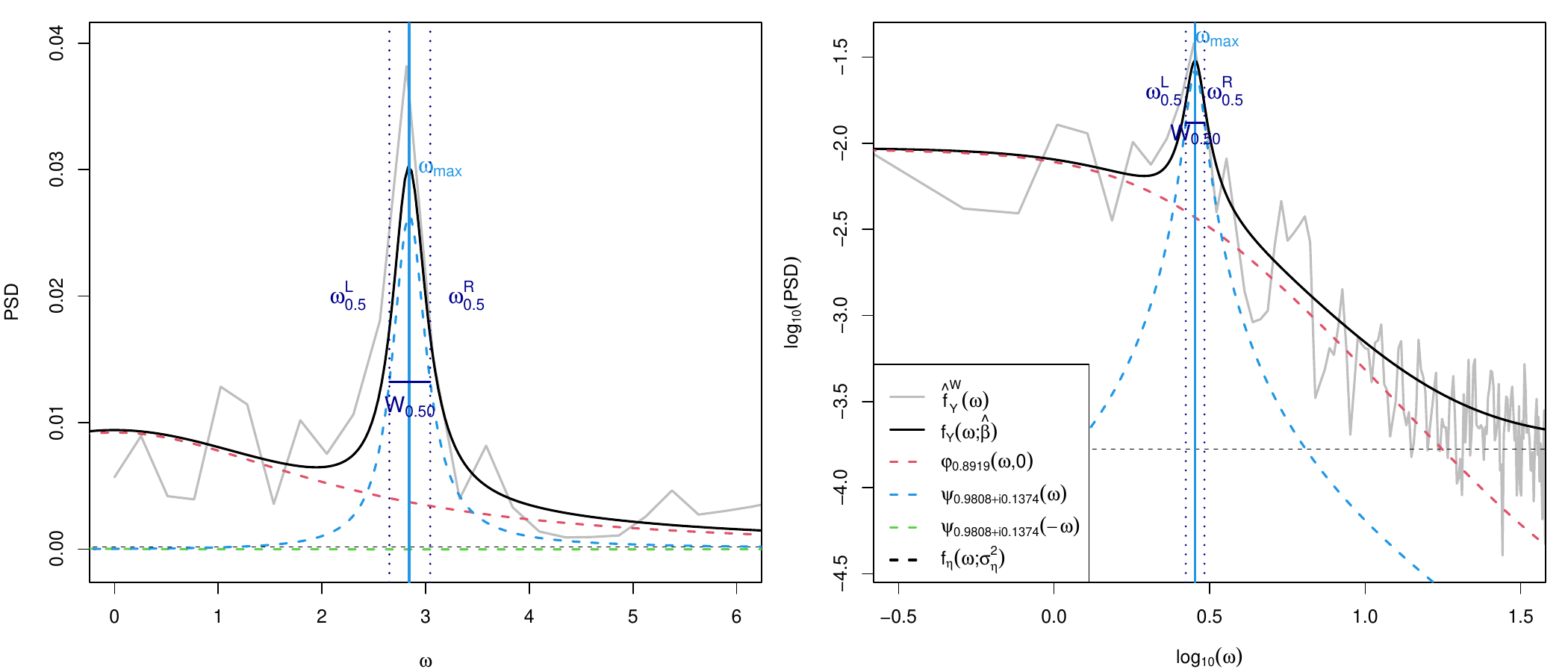}
\caption{Estimated PSD decomposition for the RXTE light curve of the X-ray binary XTE 1550$-$564 under the ARMA$(3,2)$ model. Left: PSD in the frequency domain; right: same quantities on a log-log scale. The gray line is the Welch periodogram $\hat f_Y^{W}(\omega)$ of \eqref{Welch_periodogram}; the solid black line is the fitted ARMA spectral density $f_Y(\omega;\widehat{\vect{\beta}})$; the red dashed line is the real-root component $\varphi_{0.8919}(\omega,0)$; the blue and green dashed lines are the terms $\psi_{0.9808+i0.1374}(\omega)$ and $\psi_{0.9808+i0.1374}(-\omega)$; and the black dashed line is the measurement error spectral density $f_\eta(\omega;\sigma_\eta^2)$ of \eqref{PSD_me}. The solid blue vertical line marks the estimated QPO peak frequency $\widehat{\omega}_{\max} = 2.8417$~Hz, and the dotted navy lines indicate the estimated half-power limits $\widehat{\omega}_{0.5}^{L} = 2.6511$~Hz and $\widehat{\omega}_{0.5}^{R} = 3.0445$~Hz, giving an estimated bandwidth $\widehat{W}_{0.5} = 0.1906$~Hz.}
\label{aux12}
\end{figure}

Figure~\ref{aux12} shows the PSD component decomposition of the fitted ARMA$(3,2)$ model. The autoregressive polynomial $\widehat{\phi}(z^{-1}) = 1 - 2.8535\,z^{-1} + 2.7304\,z^{-2} - 0.8748\,z^{-3}$ has three roots: one real root corresponding to the class $\widehat{\tau}_1 = \{0.8919\}$ and a complex-conjugate pair at $\widehat{\tau}_2 = \{0.9808 \pm i0.1374\}$. Following the decomposition in \eqref{PSD_decomp4}, this yields three spectral components: the real root contributes a broadband component $\varphi_{0.8919}(\omega,0)$ centered at zero frequency, while the complex-conjugate pair gives rise to the pair of components $\psi_{0.9808+i0.1374}(\pm\omega)$, which capture the structure associated with the QPO. The peak frequency of the dominant term $\psi_{0.9808+i0.1374}(\omega)$ is $\widehat{\omega}_{\max} = 2.8417$~Hz, which closely matches the peak of the Welch periodogram at $\omega_{\text{periodogram}} = 2.816$~Hz, consistent with the strong QPO reported by \cite{Kelly2014} for this source, confirming that the ARMA$(3,2)$ model localizes the dominant QPO through its complex root pair.

The same figure also reports the characterization of the dominant QPO derived from the fitted
spectral density. The estimated peak frequency is $\widehat{\omega}_{\max} = 2.8417$~Hz, and the
half-power limits are $\widehat\omega_{0.5}^{L} = 2.6511$~Hz and
$\widehat\omega_{0.5}^{R} = 3.0445$~Hz, yielding a full width at half maximum
$\widehat W_{0.5} = 0.1906$~Hz, evaluated at $k = 0.5$ of the peak amplitude. The corresponding
asymmetry coefficient is $AI_{0.5}=-0.031$, indicating a slightly asymmetric peak in which the right
half-width $\widehat W_{0.5}^{R}$ is marginally larger than the left half-width
$\widehat W_{0.5}^{L}$, so that the component decays somewhat more gradually on the
high-frequency side of $\widehat\omega_{\max}$ than on the low-frequency side. To assess the
uncertainty of these estimates, 95\% bootstrap confidence intervals were computed for
$\omega_{\max}$ and $W_{0.5}$, yielding $[2.6587, 3.1367]$~Hz and $[0.1664, 1.3397]$~Hz,
respectively. The interval for the width is markedly wider than that for the peak frequency,
spanning almost an order of magnitude. This reflects the position of the fitted
class: with $\widehat s_2 = 0.9904$ the root lies close to the unit circle, where
$W_{0.5}$ in \eqref{eq:width} varies rapidly with the modulus, so moderate
variation in $\widehat s_2$ across bootstrap replications translates into large
variation in the estimated width. The peak frequency, which depends chiefly on
the argument, is far better determined.

\begin{table}[b]
\centering
\caption{Variance contribution of each class of the fitted ARMA$(3,2)$ model for the RXTE
light curve of the X-ray binary XTE 1550$-$564.}
\label{tab:variance_contrib_apply}
\begin{tabular}{@{}ccc@{}}
\hline
        \multicolumn{1}{c}{Class/Measurement error} & \multicolumn{1}{c}{Variance contribution} & \multicolumn{1}{c}{Percentage (\%)} \\
        \hline
        $\{0.8919\}$            & $0.0337$ & $55.4745$ \\
         $\{0.9808 \pm i0.1374\}$ & $0.0164$ & $26.8443$ \\
         $\sigma^2_{\eta}$ & $0.0107$ & $17.6812$ \\
        Total & $0.0608$ & $100$ \\
        \hline
\end{tabular}
\end{table}

The variance decomposition of the RXTE light curve of XTE 1550$-$564, obtained from the fitted
ARMA$(3,2)$ model, is summarized in Table~\ref{tab:variance_contrib_apply}. The total estimated
variance is $0.0608$, of which the real root class $\{0.8919\}$ accounts for $0.0337$
($55.47\%$), reflecting the broadband, low-frequency variability of the source; the
complex-conjugate pair $\{0.9808 \pm i0.1374\}$ accounts for $0.0164$ ($26.84\%$), capturing the
variance concentrated around the dominant QPO at $2.8417$~Hz; and the measurement error
component $\sigma^2_\eta$ accounts for the remaining $0.0107$ ($17.68\%$). This decomposition shows that broadband aperiodic variability, captured by the real root
component, is the dominant contributor to the total variance, exceeding the contribution of
the QPO component itself, which nonetheless remains a non-negligible source of variability in
the analyzed segment.

Figure~\ref{Figure:resid_APPLI} presents residual diagnostics for the fitted ARMA$(3,2)$ model. The top panel shows that the residuals fluctuate randomly around zero with no apparent trend or systematic pattern, suggesting adequate model fit. The ACF of the residuals and of the squared residuals show no significant autocorrelation beyond the 95\% confidence bounds at any lag, indicating no evidence of remaining serial dependence or conditional heteroskedasticity. The histogram and the normal Q-Q plot are consistent with the normality assumption, the latter further supported by the Kolmogorov-Smirnov test ($D = 0.0261$, $p\text{-value} = 0.5038$). Finally, the Ljung-Box $p$-values remain well above the 5\% significance level for all lags $h = 1,\ldots,30$, providing no evidence of residual serial correlation. Taken together, these diagnostics support the adequacy of the fitted ARMA$(3,2)$ specification.

\begin{figure}[!b]
\centering
\includegraphics[width=10cm]{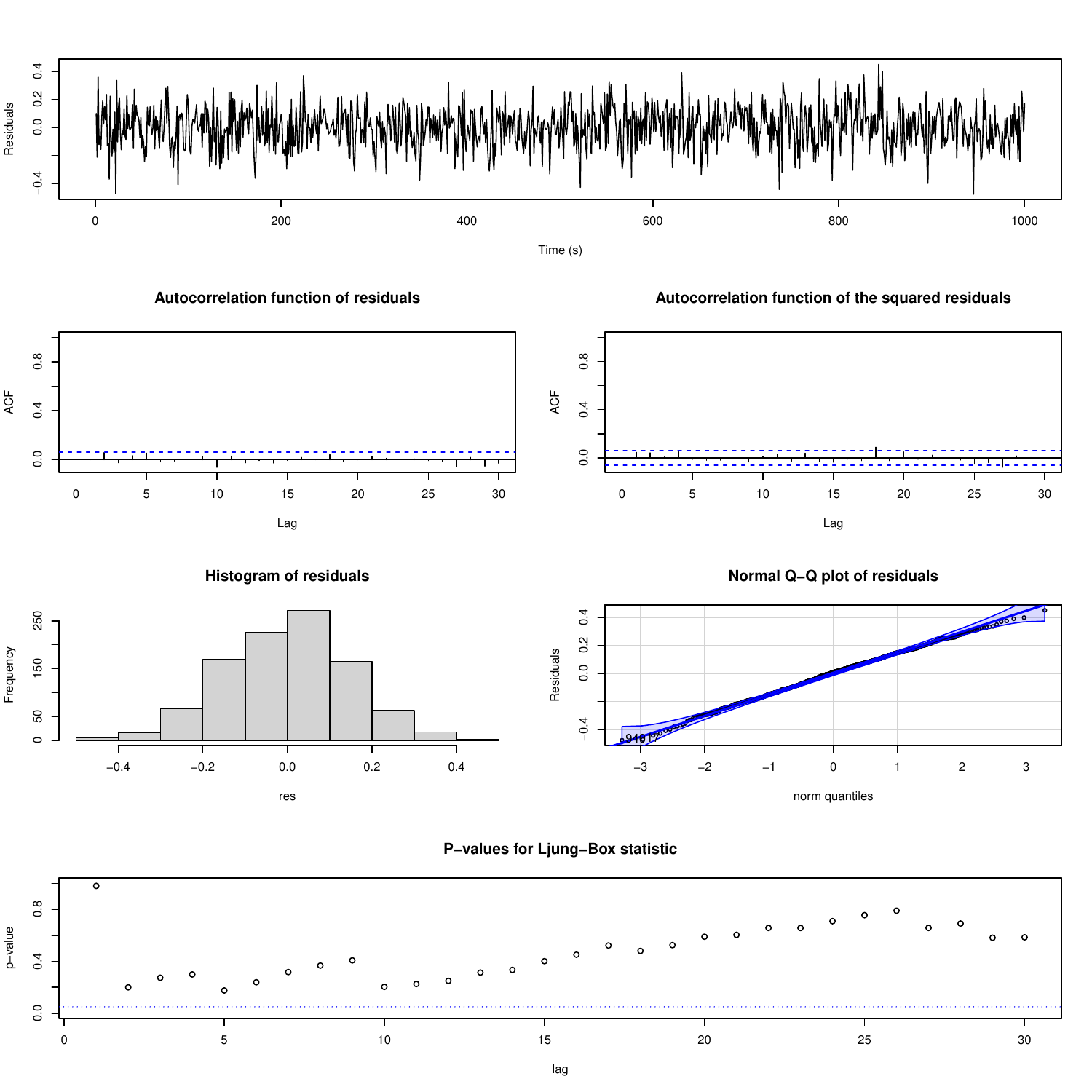}
\caption{Residual diagnostics for the fitted ARMA$(3,2)$ model applied to the RXTE light curve of XTE 1550-564. Top panel: time series of the
residuals. Middle panels: sample autocorrelation function (ACF) of the
residuals (left) and of the squared residuals (right), and histogram of
the residuals alongside the normal Q-Q plot. Bottom panel: p-values of
the Ljung-Box statistic for lags $h = 1,\ldots,30$. The blue dashed line marks the 5\% significance level. All diagnostics support the adequacy
of the fitted model.}
\label{Figure:resid_APPLI}
\end{figure}

\section{Conclusion}\label{conclusion}

This paper provides an exact decomposition of the spectral density of a
discrete-time ARMA process with measurement error into components indexed by the
equivalence classes of the autoregressive roots, together with closed-form
expressions for the location, width and asymmetry of the peak of each complex
class. The decomposition is the discrete-time counterpart of the Lorentzian
representation of a CARMA spectral density, and it extends that representation to
the ARMA processes that do not arise as sampled CARMA processes. Because the peak
descriptors are explicit functions of the model parameters, they are obtained by
substitution once the model is fitted, with no search over frequency and no
separate fit to a non-parametric spectral estimate.

The simulation study of Section~\ref{simulations} supports three conclusions.
First, the closed form of Proposition~\ref{prop:peak} recovers the peak location
without bias whenever the corresponding root is identified. Second, the nominal
centroid $\alpha/(2\pi\Delta)$, is an adequate approximation only when the component is sharply
concentrated: for a class of modulus $0.70$ it carries a bias of $0.0137$ that
does not diminish with the sample size, and across the whole design the empirical
bias of the centroid differs from the displacement predicted by
Proposition~\ref{prop:peak} by at most $0.0023$. Third, a numerical search on the
fitted spectral density is not a substitute, since the maximum of the total
density is not the maximum of an individual class component; for the same broad
class its root mean squared error exceeds that of the closed form by an order of
magnitude. The width and the asymmetry index, which have no counterpart among the
quantities read off a fitted root, are recovered with small bias and useful
precision at $n=1000$.

The application to the RXTE light curve of XTE 1550$-$564 illustrates the
framework on data. The fitted ARMA$(3,2)$ model, selected as the most parsimonious
specification reproducing the dominant oscillation, localises it at approximately
$2.8$~Hz through a single complex-conjugate root pair, and the variance
decomposition of Section~\ref{decomp_PSD_ARMA} attributes $55.5\%$ of the
total variance to the broadband component associated with the real class,
$26.8\%$ to the complex-conjugate class, and the remaining $17.7\%$ to
measurement noise. Residual diagnostics support the adequacy of the fit.

Several limitations remain, and they indicate the natural directions for further
work. The decomposition assumes the autoregressive roots to be simple; repeated
roots produce poles of higher order and components that are no longer weighted
Poisson kernels, and the weights $K(r)$ become numerically unstable as two roots
approach one another. Order selection is the main practical obstacle: for the
ARMA$(3,2)$ model of the simulation, BIC recovered the true order in $76.0\%$ of
replications without measurement error and in only $1.5\%$ with it, the
moving-average parameters being weakly identified relative to the penalty, though
the complex class was still matched in $97.6\%$ of the latter. At $n=100$ the
descriptors are too imprecise to be reported, so the framework is suited to the
long, regularly sampled series for which it was designed rather than to short
records. Uncertainty in the peak descriptors was assessed here by the parametric
bootstrap, which for the light curve gave a wide interval for the width; since
the descriptors are smooth functions of the model parameters away from repeated
roots, standard errors also follow by the delta method from the observed
information, and calibrating either route is left for future work.
Finally, the measurement error variance has been treated as known throughout,
estimated externally in the application \citep{Vaughan2003}; its joint estimation
with the ARMA parameters, and the identifiability of that joint problem, deserve
separate study, as does the behaviour of the method when the data-generating
process is not exactly ARMA.

\section*{Acknowledgments}
The authors would like to thank the anonymous referees, an Associate
Editor and the Editor for their constructive comments that improved the
quality of this paper.

\section*{Funding}
This work was funded by ANID/FONDECYT/INICIACI\'{O}N/N.11240804.

\section*{Supplementary material and code}
The supplement to this paper provides the full Kalman filter recursions used for
likelihood evaluation of the ARMA state-space model (S.1), the parametric bootstrap
procedure used to construct standard deviations and confidence intervals for the
estimated model parameters (S.2), and the R code implementing the proposed methodology,
including the simulation study and the application to the XTE 1550-564 light curve.
Code for replicating the results presented in this paper is available on GitHub:
\url{https://github.com/DarlinSoto/PSD_decomposition_ARMA}.

\appendix
\section{Proofs}

The Appendix contains the proofs of Propositions~\ref{prop:PSD_ARMA},
\ref{prop:peak}, \ref{prop:length} and \ref{prop:AIclosed}, and of
equation~\eqref{Pcomplex}, together with a level-set lemma on which the last
three proofs rest.

\subsection{Proof of Proposition \ref{prop:PSD_ARMA}}\label{proof_prop1}

For a causal ARMA($p,q$) process, \cite{Brockwell2001} showed that if the roots
$\mathcal{R} = \{r_1, \ldots, r_p\}$ of
$\phi(z^{-1}) = 1 - \phi_1 z^{-1} - \cdots - \phi_p z^{-p}$ are distinct, then the
autocovariance function can be written as
\begin{equation}\label{ACF_ARMA}
\gamma_X(h;\vect{\beta})=\sigma^2_Z \left\{ \frac{ \mathbb{I}_{[0, q-p]}(|h|)}{(q-p-|h|)!} \left[ \frac{d^{q-p-|h|}}{dz^{q-p-|h|}} \frac{z^{q-p} \theta(z) \theta(z^{-1})}{\phi(z) \phi(z^{-1})} \right]_{z=0} -\sum_{r \in \mathcal{R}}  \frac{r^{|h|+1} \theta(r) \theta(r^{-1})}{\phi(r) \phi'(r^{-1})} \right \}.
\end{equation}
Here and throughout, $\phi'(r^{-1})$ denotes the derivative of $\phi$ evaluated at
$r^{-1}$, not the derivative of $z\mapsto\phi(z^{-1})$ at $r$. On the other hand,
the PSD of a zero-mean stationary time series is given by
\begin{equation}\label{PSD_def}
\PSD_X(\omega;\vect{\beta})=\frac{\Delta}{2\pi} \sum^{\infty}_{h=-\infty} e^{-ih2\pi \omega \Delta}\gamma_X(h;\vect{\beta}).
\end{equation}
Substituting \eqref{ACF_ARMA} into \eqref{PSD_def}, we obtain
\begin{align}
\PSD_X(\omega;\vect{\beta}) =&\frac{\Delta \sigma^2_Z}{2\pi} \sum^{\infty}_{h=-\infty} e^{-ih2\pi \omega \Delta}\frac{ \mathbb{I}_{[0, q-p]}(|h|)}{(q-p-|h|)!} \left[ \frac{d^{q-p-|h|}}{dz^{q-p-|h|}} \frac{z^{q-p} \theta(z) \theta(z^{-1})}{\phi(z) \phi(z^{-1})} \right]_{z=0} \label{PSD_0}\\
&-\frac{\Delta \sigma^2_Z}{2\pi} \sum^{\infty}_{h=-\infty} e^{-ih2\pi \omega \Delta}\sum_{r \in \mathcal{R}}  \frac{r^{|h|+1} \theta(r) \theta(r^{-1})}{\phi(r) \phi'(r^{-1})}. \label{PSD_rr}
\end{align}
We now develop the first term, given by \eqref{PSD_0}, which is present only when
$q\geq p$; let $m=q-p$. The indicator $\mathbb{I}_{[0,m]}(|h|)$ vanishes for
$|h|>m$, so the sum reduces to the finite range $h=-m,\dots,m$:
$$\frac{\Delta \sigma^2_Z}{2\pi} \sum_{h=-m}^{m} c_h\,e^{-ih2\pi\omega\Delta},$$
where
$$c_h = \frac{1}{(m-|h|)!}\left[\frac{d^{\,m-|h|}}{dz^{\,m-|h|}}\,
\frac{z^{q-p}\,\theta(z)\theta(z^{-1})}{\phi(z)\phi(z^{-1})}\right]_{z=0},\quad |h|\le m.$$
Separating the $h=0$ term and pairing $h$ with $-h$, this sum can be written as
\begin{eqnarray*}
  \frac{\Delta \sigma^2_Z}{2\pi} \sum_{h=-m}^{m} c_h\,e^{-ih2\pi\omega\Delta} & = & \frac{\Delta \sigma^2_Z}{2\pi} \left [ c_0 + \sum_{h=1}^{m} \left \{  c_h e^{-ih2\pi\omega\Delta} +  c_{-h} e^{ih2\pi\omega\Delta} \right \}  \right ].
\end{eqnarray*}
Since $c_h$ depends on $h$ only through $|h|$, we have $c_{-h}=c_h$, and the above
equation simplifies to
\begin{eqnarray*}
  \frac{\Delta \sigma^2_Z}{2\pi} \sum_{h=-m}^{m} c_h\,e^{-ih2\pi\omega\Delta} & = & \frac{\Delta \sigma^2_Z}{2\pi} \left [ c_0 + \sum_{h=1}^{m} c_h\left \{   e^{-ih2\pi\omega\Delta} + e^{ih2\pi\omega\Delta} \right \}  \right ] \nonumber \\
  & = & \frac{\Delta \sigma^2_Z}{2\pi} \left [ c_0 + 2\sum_{h=1}^{m} c_h \cos(2\pi h \omega \Delta)  \right ],
\end{eqnarray*}
which corresponds to the term $\PSD_{X,0}(\omega;\vect{\beta})$ in \eqref{P0}.

Regarding the second term, given by \eqref{PSD_rr}, the sum over $r\in\mathcal{R}$
is finite. Therefore, we can interchange the order of summation and write
$$-\frac{\Delta \sigma^2_Z}{2\pi} \sum_{r \in \mathcal{R}} \frac{r \theta(r) \theta(r^{-1})}{\phi(r) \phi'(r^{-1})} \sum^{\infty}_{h=-\infty} e^{-ih2\pi \omega \Delta} r^{|h|}.$$
The inner summation $\sum^{\infty}_{h=-\infty} e^{-ih2\pi \omega \Delta} r^{|h|}$ can be
computed as
\begin{flalign*}
\sum^{\infty}_{h=-\infty} e^{-ih2\pi \omega \Delta} r^{|h|}&=\sum^{0}_{h=-\infty} e^{-ih2\pi \omega \Delta} r^{-h} + \sum^{\infty}_{h=0} e^{-ih2\pi \omega \Delta} r^{h} - 1\\
&=  \sum^{\infty}_{h=0} (e^{i2\pi \omega \Delta} r)^{h} + \sum^{\infty}_{h=0} (e^{-i2\pi \omega \Delta} r)^{h} - 1.
\end{flalign*}
Since $|r| < 1$, we can apply the geometric series formula
\begin{flalign*}
\sum^{\infty}_{h=-\infty} e^{-ih2\pi \omega \Delta} r^{|h|}&=\frac{1}{1-r e^{i2\pi \omega \Delta}} + \frac{1}{1-r e^{-i2\pi \omega \Delta}} -1\\
&= \frac{1-r^2}{1-2r \cos \left(2\pi \omega \Delta \right)+r^2}.
\end{flalign*}
Therefore, the second component of the PSD in \eqref{PSD_rr} can be rewritten as
$$\sum_{r \in \mathcal{R}} K(r)\,\frac{\Delta(1-r^2)}{1-2r \cos(2\pi \omega \Delta)+r^2},$$
whose summands are the components $\PSD_{X,r}(\omega;\vect{\beta})$ of \eqref{Pr}.

\subsection{Proof of Equation \eqref{Pcomplex}}\label{proof_P_complex}

The component $\PSD_{X,r}(\omega;\vect{\beta})$ defined by \eqref{Pr} can also be expressed as
\begin{equation}\label{psd_aux}
\PSD_{X,r}(\omega;\vect{\beta})= - \frac{\Delta \sigma^2_Z}{2\pi} \frac{r \theta(r) \theta(r^{-1})} {\phi(r) \phi'(r^{-1})} \sum_{h=-\infty}^{\infty} e^{-i h 2\pi \omega \Delta}\, r^{|h|}.
\end{equation}
Consider the case of complex roots, that is, $r = a+ib= s\,e^{i\alpha}$ and
$\overline{r} = a-ib= s\,e^{-i\alpha}$, where $s = \sqrt{a^2 + b^2}=|r|$ and
$\alpha = \arg(r)$.

We first replace $r$ in \eqref{psd_aux} by the root $r = s\,e^{i\alpha}$, obtaining
\begin{align*}
\PSD_{X,{r}}(\omega;\vect{\beta})
&= -\frac{\Delta\sigma^2_Z}{2\pi} \frac{r \theta(r) \theta(r^{-1})}{\phi(r)\, \phi'(r^{-1})} \sum_{h=-\infty}^{\infty} e^{-i h 2\pi \omega \Delta}\, r^{|h|}\\
&=-\frac{\Delta\sigma^2_Z}{2\pi} \frac{r \theta(r) \theta(r^{-1})}{\phi(r)\, \phi'(r^{-1})}  \left\{\sum_{h=0}^{\infty} e^{i h 2\pi \omega \Delta}\, r^h +\sum_{h=0}^{\infty} e^{-i h 2\pi \omega \Delta}\, r^{h}-1\right\}\\
&= \Delta K(r) \left\{ \sum_{h=0}^{\infty} \left( s e^{i(2\pi \omega \Delta+\alpha)} \right)^h+\sum_{h=0}^{\infty} \left( s e^{-i(2\pi \omega \Delta-\alpha)} \right)^h- 1 \right\} \nonumber\\
&= \Delta K(r) \left\{ \frac{1}{1 - s e^{i(2\pi \omega \Delta+\alpha)}}
   + \frac{1}{1 - s e^{-i(2\pi \omega \Delta-\alpha)}} - 1 \right\} \nonumber\\
&=\Delta \left[ \Real(K(r)) + i \Imag(K(r)) \right]
   \left\{ \vphantom{\frac{1}{1 - s e^{i(2\pi \omega \Delta+\alpha)}}} \frac{1}{1 - s e^{i(2\pi \omega \Delta+\alpha)}}+
   \frac{1}{1 - s e^{-i(2\pi \omega \Delta-\alpha)}}- 1\right\}
\end{align*}
where
$$K(r)=-\frac{\sigma^2_Z}{2\pi} \frac{r\,\theta(r)\, \theta\big(r^{-1}\big)}{\phi(r)\, \phi'\big(r^{-1}\big)}.$$

Next, we repeat the procedure replacing $r$ in \eqref{psd_aux} by the root
$\overline{r}=s\,e^{-i\alpha}$, obtaining
\begin{align*}
\PSD_{X,\overline{r}}(\omega;\vect{\beta})
&=\Delta \left[ \Real(K(\overline{r})) + i \Imag(K(\overline{r})) \right]
   \left\{ \vphantom{\frac{1}{1 - s e^{i(2\pi \omega \Delta-\alpha)}}} \frac{1}{1 - s e^{i(2\pi \omega \Delta-\alpha)}}
   + \frac{1}{1 - s e^{-i(2\pi \omega \Delta+\alpha)}} - 1 \right\}
\end{align*}
where
$$K(\overline{r})=-\frac{\sigma^2_Z}{2\pi} \frac{\overline{r}\, \theta(\overline{r})\, \theta\big(\overline{r}^{-1}\big)}{\phi(\overline{r})\, \phi'\big(\overline{r}^{-1}\big)}.$$

Since $\phi$ and $\theta$ have real coefficients, we have $K(\overline{r})=\overline{K(r)}$. Thus,
$$\PSD_{X,\overline{r}}(\omega;\vect{\beta})= \Delta \left[ \Real(K(r)) - i \Imag(K(r)) \right] \left\{ \frac{1}{1 - s e^{i(2\pi \omega \Delta-\alpha)}}+ \frac{1}{1 - s e^{-i(2\pi \omega \Delta+\alpha)}} - 1 \right\}.$$

Summing the terms $\PSD_{X,{r}}(\omega;\vect{\beta})$ and $\PSD_{X,{\overline{r}}}(\omega;\vect{\beta})$ yields
\begin{align*}
    &\widetilde{\PSD}_{X,\tau}(\omega;\vect{\beta})=\PSD_{X,{r}}(\omega;\vect{\beta})+\PSD_{X,{\overline{r}}}(\omega;\vect{\beta})\\
    &= \Delta \left[ \Real(K(r)) + i \Imag(K(r)) \right]
   \left\{ \frac{1}{1 - s e^{i(2\pi \omega \Delta+\alpha)}}+
   \frac{1}{1 - s e^{-i(2\pi \omega \Delta-\alpha)}}- 1\right\} \\
   &+ \Delta  \left[ \Real(K(r)) - i \Imag(K(r)) \right]
   \left\{ \frac{1}{1 - s e^{i(2\pi \omega \Delta-\alpha)}}+
   \frac{1}{1 - s e^{-i(2\pi \omega \Delta+\alpha)}}- 1\right\}\\
   &= \Delta  \Real(K(r))  \left[  \frac{1}{1 - s e^{i(2\pi \omega \Delta+\alpha)}} + \frac{1}{1 - s e^{-i(2\pi \omega \Delta+\alpha)}} - 1 \right]\\
   &+ i \Delta \Imag(K(r))  \left[  \frac{1}{1 - s e^{i(2\pi \omega \Delta+\alpha)}} - \frac{1}{1 - s e^{-i(2\pi \omega \Delta+\alpha)}} \right]\\
   &+ \Delta  \Real(K(r))  \left[  \frac{1}{1 - s e^{i(2\pi \omega \Delta-\alpha)}} + \frac{1}{1 - s e^{-i(2\pi \omega \Delta-\alpha)}} - 1 \right]\\
   &+ i \Delta \Imag(K(r)) \left[  \frac{1}{1 - s e^{i(2\pi \omega \Delta-\alpha)}} - \frac{1}{1 - s e^{-i(2\pi \omega \Delta-\alpha)}} \right]\\
   &=\Real(K(r)) \frac{\Delta (1-s^2)}{1-2s\cos(2\pi \omega \Delta + \alpha) + s^2} -  \Imag(K(r))  \frac{\Delta  2s \sin(2\pi \omega \Delta+\alpha)}{1-2s\cos(2\pi \omega \Delta + \alpha) + s^2}\\
   &+ \Real(K(r)) \frac{\Delta (1-s^2)}{1-2s\cos(2\pi \omega \Delta - \alpha) + s^2} +\Imag(K(r)) \frac{\Delta  2s \sin(2\pi \omega \Delta-\alpha)}{1-2s\cos(2\pi \omega \Delta - \alpha) + s^2}.
\end{align*}

Finally, defining the component functions
\begin{align*}
    \varphi_{s}(\omega,\alpha)&= \Real\!\big(K(r)\big) \frac{\Delta (1-s^2)}{1-2s\cos(2\pi \omega \Delta - \alpha)+s^2},\\
    \zeta_{s}(\omega,\alpha)&= \Imag\!\big(K(r)\big) \frac{\Delta\, 2s \sin(2\pi \omega \Delta-\alpha)}{1-2s\cos(2\pi \omega \Delta - \alpha) + s^2},\\
    \varphi_{s}(-\omega,\alpha)&= \Real\!\big(K(r)\big) \frac{\Delta (1-s^2)}{1-2s\cos(2\pi \omega \Delta + \alpha) + s^2},\\
    \zeta_{s}(-\omega,\alpha)&= -\Imag\!\big(K(r)\big) \frac{\Delta\, 2s \sin(2\pi \omega \Delta+\alpha)}{1-2s\cos(2\pi \omega \Delta + \alpha) + s^2},\\
    \psi_{\tau}(\omega) &= \varphi_{s}(\omega,\alpha) + \zeta_{s}(\omega,\alpha),\\
    \psi_{\tau}(-\omega) &= \varphi_{s}(-\omega,\alpha) + \zeta_{s}(-\omega,\alpha),
\end{align*}
we obtain the expression in \eqref{Pcomplex}.

\subsection{A level-set lemma}\label{proof_lemma_arc}

Throughout the remainder of this appendix, $\tau=\{r,\bar r\}$ is a
complex-conjugate class with $r=s\,e^{i\alpha}$, $s\in(0,1)$, and we write
$A=\Real\big(K(r)\big)$, $B=\Imag\big(K(r)\big)$ and $|K|=\sqrt{A^2+B^2}$,
assuming $K(r)\neq0$. Setting
\begin{equation}\label{app:lambda}
\lambda = 2\pi\omega\Delta-\alpha,
\qquad
D(\lambda)=1+s^2-2s\cos\lambda,
\qquad
N(\lambda)=A(1-s^2)+2Bs\sin\lambda,
\end{equation}
the component \eqref{Pcomplex_LS} reads $\psi_\tau(\omega)=\Delta\,g(\lambda)$ with
$g=N/D$. Since $D(\lambda)\ge(1-s)^2>0$, the function $g$ is smooth, bounded and
$2\pi$-periodic, and $\omega\mapsto\lambda$ maps the frequency band
$(-1/(2\Delta),1/(2\Delta)]$ bijectively onto the circle, so every statement about
$g$ on the circle transfers to $\psi_\tau$ on the band.

\begin{lemma}\label{lem:arc}
Fix $\kappa\in\mathbb{R}$ and set
\begin{equation}\label{app:PQC}
C(\kappa)=\kappa(1+s^2)-A(1-s^2),
\qquad
R(\kappa)=2s\sqrt{\kappa^2+B^2},
\qquad
\mu_\kappa=\operatorname{atan2}(B,\kappa).
\end{equation}
Then $g(\lambda)\ge\kappa$ if and only if
$R(\kappa)\cos(\lambda-\mu_\kappa)\ge C(\kappa)$. Consequently, if
$|C(\kappa)|\le R(\kappa)$ and
$\gamma_\kappa=\arccos\big(C(\kappa)/R(\kappa)\big)\in[0,\pi]$, the superlevel set
$\{\lambda:g(\lambda)\ge\kappa\}$ is the closed arc
$[\mu_\kappa-\gamma_\kappa,\ \mu_\kappa+\gamma_\kappa]$ modulo $2\pi$; it is empty
if $C(\kappa)>R(\kappa)$ and the whole circle if $C(\kappa)<-R(\kappa)$.
\end{lemma}

\begin{proof}
As $D>0$, the inequality $g\ge\kappa$ is equivalent to $N-\kappa D\ge0$, that is to
$2\kappa s\cos\lambda+2Bs\sin\lambda\ge C(\kappa)$. Writing the left-hand side as
$R(\kappa)\cos(\lambda-\mu_\kappa)$ with $R$ and $\mu_\kappa$ as in \eqref{app:PQC},
and noting that $\cos(\lambda-\mu_\kappa)\ge C/R$ holds exactly for
$|\lambda-\mu_\kappa|\le\arccos(C/R)$ modulo $2\pi$, gives the claim. The angle
$\mu_\kappa$ is undefined only if $R(\kappa)=0$, that is if $\kappa=0$ and $B=0$, in
which case $C=-A(1-s^2)$ and the last two clauses apply directly.
\end{proof}

Since every superlevel set of $g$ is a single arc, $g$ is unimodal on the circle:
it has exactly one local maximum and one local minimum. The three proofs below rest
on this.

\subsection{Proof of Proposition \ref{prop:peak}}\label{proof_prop2}

\begin{proof}
By Lemma~\ref{lem:arc} the level $\kappa$ is attained by $g$ if and only if
$C(\kappa)^2\le R(\kappa)^2$. Expanding and using $(1+s^2)^2-4s^2=(1-s^2)^2$, this is
$\Pi(\kappa)\le0$ with
\begin{equation}\label{app:quadratic}
\Pi(\kappa)=(1-s^2)^2\kappa^2-2A(1-s^4)\,\kappa+A^2(1-s^2)^2-4s^2B^2 .
\end{equation}
The leading coefficient is positive and the reduced discriminant of $\Pi$ is
$4s^2(1-s^2)^2|K|^2>0$, so $\Pi$ has the distinct real roots
\begin{equation*}
g_{\max}=\frac{A(1+s^2)+2s|K|}{1-s^2},
\qquad
g_{\min}=\frac{A(1+s^2)-2s|K|}{1-s^2},
\end{equation*}
and $\{\Pi\le0\}=[g_{\min},g_{\max}]$ is exactly the range of $g$. Multiplying by
$\Delta$ gives $G_{\max}=\Delta\,g_{\max}$ and $G_{\min}=\Delta\,g_{\min}$ as in
\eqref{eq:extremes}.

At $\kappa=g_{\max}$ we have $|C|=R$. If $C=-R$, then $\gamma_{g_{\max}}=\pi$ and
Lemma~\ref{lem:arc} would make the set the whole circle, so
$g\equiv g_{\max}$ and hence $g_{\max}=g_{\min}$, contradicting $2s|K|>0$.
Therefore $C=R$, so
$\gamma_{g_{\max}}=0$ and the arc collapses to the single point
\begin{equation*}
\lambda^\star=\mu_{g_{\max}}=\operatorname{atan2}\big(B,\,g_{\max}\big),
\end{equation*}
which is the unique maximiser. Undoing \eqref{app:lambda} through
$\omega=(\lambda+\alpha)/(2\pi\Delta)$ and reducing modulo $1/\Delta$ into the band
yields \eqref{eq:wmax}. Applying the same argument to $-g$, whose parameters are
$-A$ and $-B$, gives $C=-R$ at $\kappa=g_{\min}$, so the minimum is attained at the
single point $\lambda_{\min}=\mu_{g_{\min}}+\pi
=\operatorname{atan2}\big(-(1-s^2)B,\ 2s|K|-(1+s^2)A\big)$, which is \eqref{eq:wmin}.
\end{proof}

\subsection{Proof of Proposition \ref{prop:length}}\label{proof_prop3}

\begin{proof}
Let $\kappa_k=\ell_k/\Delta=k\,g_{\max}$. By Proposition~\ref{prop:peak} the level
$\ell_k$ is attained if and only if $g_{\min}\le\kappa_k\le g_{\max}$. Since
$G_{\max}>0$, the upper bound holds for every $k\le1$, and the lower bound reads
$k\ge g_{\min}/g_{\max}=k_{\min}$, which is the hypothesis. By Lemma~\ref{lem:arc}
this is the same as $|C(\kappa_k)|\le R(\kappa_k)$, so $\gamma_k$ in \eqref{eq:arc}
is well defined, and $\gamma_k\in(0,\pi)$ because $g_{\min}<\kappa_k<g_{\max}$.

The level set $g=\kappa_k$ is $R(\kappa_k)\cos(\lambda-\mu_k)=C(\kappa_k)$ with
$\mu_k=\mu_{\kappa_k}$, whose two solutions on the circle are
$\lambda^{L,R}_k=\mu_k\mp\gamma_k$, and the superlevel set is the arc between them.
Undoing \eqref{app:lambda} gives \eqref{eq:wLR}, and
$\lambda^R_k-\lambda^L_k=2\gamma_k$ gives the width
$W_k=2\gamma_k/(2\pi\Delta)=\gamma_k/(\pi\Delta)$ of \eqref{eq:width}. Finally
$g(\lambda^\star)=g_{\max}\ge\kappa_k$, so $\lambda^\star$ belongs to the superlevel
arc, that is, $\omega_{\max}$ lies between $\omega^L_k$ and $\omega^R_k$.
\end{proof}

\subsection{Proof of Proposition \ref{prop:AIclosed}}\label{proof_prop4}

\begin{proof}
By the proof of Proposition~\ref{prop:length},
$2\pi\Delta\,W^L_k=(\lambda^\star-\mu_k)+\gamma_k$ and
$2\pi\Delta\,W^R_k=\gamma_k-(\lambda^\star-\mu_k)$, so
$AI_k=(\lambda^\star-\mu_k)/\gamma_k$. Since
$\operatorname{atan2}(B,x)=\arg(x+iB)$ and $\kappa_k=k\,g_{\max}$,
\begin{equation*}
\lambda^\star-\mu_k=\arg\frac{g_{\max}+iB}{\kappa_k+iB}
=\arg\Big[k\,g_{\max}^{2}+B^{2}-i(1-k)B\,g_{\max}\Big],
\end{equation*}
whose real part is positive, so the argument lies in $(-\pi/2,\pi/2)$ and equals the
$\arctan$ in \eqref{eq:AIclosed}. The sign statement and the case $B=0$ follow at
once. As $k\uparrow1$ the numerator is $O(1-k)$ while $\gamma_k=O(\sqrt{1-k})$,
because $g_{\max}$ is a simple root of $\Pi$ in \eqref{app:quadratic}, whence
$AI_k=O(\sqrt{1-k})$.
\end{proof}

\clearpage
\setcounter{section}{0}
\renewcommand{\thesection}{S.\arabic{section}}
\section*{Supplementary Material}

This supplement provides the full Kalman filter recursions used for likelihood
evaluation of the ARMA state-space model, and the parametric bootstrap procedure, used to construct standard deviations and confidence intervals for
the estimated model parameters.

\section{Kalman filter recursions}
\label{S1}

This section details the Kalman filter recursions used to evaluate the
likelihood of the ARMA$(p,q)$ state-space model with measurement error. 

The following steps summarize the procedure.

\begin{enumerate}
    \item Initialize the state vector $\vect{x}$ and its covariance matrix $\matr{P}$ at time $1$ as
    \begin{align*}
    \vect{x}_{1|0} &= \vect{x}^*_{1|0}\\
    \matr{P}_{1|0} &= \matr{P}^*_{1|0},
    \end{align*}
    where $\vect{x}^*_{1|0}$ and $\matr{P}^*_{1|0}$ are specified initial values. 
    
    \item Calculate the prediction of the next observation as  $$\mean({y}_1|{y}_0)=\vect{F}\vect{x}_{1|0}.$$
    
    \item The innovation, $r_{1}$, is the predictive residual, i.e.,
    $$r_{1}={y}_{1}-\mean({y}_1|{y}_0).$$

    \item The innovation variance of ${y}_{1}$ is
    $$\vari({y}_1|{y}_0)=\vect{F}\matr{P}_{1|0}\vect{F}^\top+\sigma^2_{\eta}.$$

    \item Initialize the Kalman gain:
    $$\vect{K}_{1}=\matr{P}_{1|0}\vect{F}^{\top} \vari({y}_1|{y}_0)^{-1}.$$

    \item Update the estimate of the state vector
    $$\vect{x}_{1|1}=\vect{x}_{1|0}+\vect{K}_{1}r_{1}$$
    with the corresponding covariance matrix being
    $$\matr{P}_{1|1}=\matr{P}_{1|0}-\vari({y}_1|{y}_0)\vect{K}_{1}\vect{K}^{\top}_{1}.$$

After initializing the Kalman filter as above, repet the following steps for $t=2,\dots,n$:

    \item Calculate the one-step ahead predictor
$$\vect{x}_{t|t-1}=\matr{G}\vect{x}_{t-1|t-1}.$$

    \item Calculate the covariance matrix of the preceding predictor
$$\matr{P}_{t|t-1}=\matr{G}\matr{P}_{t-1|t-1}\matr{G}^{\top}+\matr{\Sigma}_w.$$

    \item The predictor of the next observation is 
$$\mean({y}_{t}|\vect{{y}}_{<t})=\vect{F}\vect{x}_{t|t-1}.$$

    \item The innovation, $r_{t}$, is the predictive residual, i.e.,
$$r_{t}={y}_{t}-\mean({y}_{t}|\vect{{y}}_{<t}).$$

    \item The innovation variance of ${y}_{t}$ is
$$\vari({y}_{t}|\vect{{y}}_{<t})=\vect{F}\matr{P}_{t|t-1}\vect{F}^{\top}+\sigma^2_{\eta}.$$

    \item The Kalman gain matrix is
$$\vect{K}_{t}=\matr{P}_{t|t-1}\vect{F}^{\top} \vari({y}_{t}|\vect{{y}}_{<t})^{-1}.$$

    \item Update the estimate of the state vector
$$\vect{x}_{t|t}=\vect{x}_{t|t-1}+\vect{K}_{t}r_{t}$$
with the corresponding covariance matrix being
$$\matr{P}_{t|t}=\matr{P}_{t|t-1}-\vari({y}_{t}|\vect{{y}}_{<t})\vect{K}_{t}\vect{K}^{\top}_{t}.$$

\end{enumerate}

The log-likelihood function of the observations $\vect{{y}}=({y}_{1},\dots,{y}_{n})^\top$ is
    $$L(\sigma^2_Z, \sigma^2_{\eta}, \vect{\phi}, \vect{\theta}; \vect{{y}} )=-\frac{n}{2} \log(2\pi)-\frac{1}{2} \sum^{n}_{t=1} \log[\vari({y}_{t}|\vect{{y}}_{<t})] - \frac{1}{2} \sum^{n}_{t=1} \frac{r^2_{t}}{\vari({y}_{t}|\vect{{y}}_{<t})}.$$


\section{Parametric bootstrap procedure}
\label{S2}

This section describes the parametric bootstrap procedure for the state-space
representation of the ARMA$(p,q)$ model with measurement error, used to assess the uncertainty of the parameter estimates of the
selected model.

The procedure consists of the following steps.

\begin{enumerate}
    \item Compute the standardized innovations from the fitted model:
    \begin{equation*}
        e_t^*(\widehat{\boldsymbol{\tau}}) = 
        \frac{r_t}{\vari({y}_{t}|\vect{{y}}_{<t})^{1/2}},
    \end{equation*}
    where $r_t = y_t - \mean({y}_{t}|\vect{{y}}_{<t})$ and $\vari({y}_{t}|\vect{{y}}_{<t})=\vect{F}\matr{P}_{t|t-1}\vect{F}^{\top}+\sigma^2_{\eta}$ 
    are the innovations and their variances, respectively, obtained by running the 
    Kalman filter recursions.
    
    \item Sample with replacement $n$ times from the set 
    $\{e_1^*(\widehat{\boldsymbol{\tau}}), \ldots, 
    e_n^*(\widehat{\boldsymbol{\tau}})\}$ to obtain a bootstrap sample 
    of standardized innovations.

    \item Construct a bootstrap data set $\{y_1^*, \ldots, y_n^*\}$ 
    using the state space representation. Define the augmented vector 
    $\vect{\xi}_t = (\vect{x}_t^\top, y_t)^\top$ which satisfies the recursion
    \begin{equation*}
        \vect{\xi}_{t+1} = \matr{L}_t \vect{\xi}_t + 
        \matr{H}_t\, e_t^*(\widehat{\boldsymbol{\tau}}),
    \end{equation*}
    where
    \begin{equation*}
        \matr{L}_t = \begin{bmatrix} \matr{G} & \mathbf{0} \\ \vect{F} & 0 \end{bmatrix}, 
        \quad
        \matr{H}_t = \begin{bmatrix} \matr{G} \vect{K}_t \,\vari({y}_{t}|\vect{{y}}_{<t})^{1/2} \\ 
        \vari({y}_{t}|\vect{{y}}_{<t})^{1/2} \end{bmatrix},
    \end{equation*}
    with $\matr{G}$ the transition matrix and $\vect{F}$ the observation matrix, and $K_t$ the Kalman gain. Initial conditions of the Kalman filter remain fixed at their given values, and the parameter vector is held fixed fixed at $\widehat{\boldsymbol{\tau}}$.

    The bootstrap 
    observations $y_t^*$ are then extracted as the last component of 
    $\boldsymbol{\xi}_t$ at each time step.

    \item Using the bootstrap data set $\{y_t^*;\, t = 1, \ldots, n\}$, 
    compute the log-likelihood $\ell(\boldsymbol{\tau}; \boldsymbol{y}^*)$ 
    via the Kalman filter recursions and obtain the maximum likelihood estimator
    \begin{equation*}
        \widehat{\boldsymbol{\tau}}^* = \operatorname*{argmax}_{\boldsymbol{\tau}}\, 
        L(\boldsymbol{\tau}; \boldsymbol{y}^*).
    \end{equation*}

    \item Repeat steps 2.-4. a total of $B$ times, obtaining a bootstrapped set of parameter estimates $\{\widehat{\boldsymbol{\tau}}_b^*;\, b = 1, \ldots, B\}$. The bootstrap standard deviation of each parameter estimate is then  computed as the empirical standard deviation of $\{\widehat{\boldsymbol{\tau}}_b^*;\, b = 1, \ldots, B\}$, and the $95\%$ confidence interval is constructed as the interval between the $2.5$th and $97.5$th empirical percentiles of the bootstrap distribution.
\end{enumerate}

\bibliographystyle{plainnat}
\bibliography{Ref}

\end{document}